\documentclass{article}
\usepackage{graphicx} 

\usepackage{amsmath,amssymb,amsthm}
\usepackage{bbm}

\usepackage{algorithm}
\usepackage{algpseudocode}
\usepackage{appendix}

\usepackage{booktabs} 
\usepackage{makecell}

\newcommand{\from}{\leftarrow}
\newcommand{\est}{\textsf{Strong Estimator }}
\newcommand{\estt}{\textsf{Strong Estimator}}

\DeclareMathOperator{\med}{median}
\newcommand{\ind}{\textsf{Indyk's p-stable Sketch }}
\newcommand{\indd}{\textsf{Indyk's p-stable Sketch}}
\newcommand{\dif}{\textsf{Difference Estimator }}
\newcommand{\diff}{\textsf{Difference Estimator}}

\newcommand{\R}{\mathbb R}

\newcommand{\eps}{\varepsilon}

\usepackage{algorithm}
\usepackage{etoolbox}

\usepackage[a4paper, total={6in, 8in}]{geometry}

\usepackage[T1]{fontenc}
\usepackage[latin9]{inputenc}
\usepackage{titlesec}
\usepackage{titling}
\usepackage{verbatim}
\usepackage{float}
\usepackage{mdframed}
\usepackage{amssymb,amsfonts,amsmath}
\usepackage{enumerate}
\usepackage{enumitem}
\usepackage{amsthm}
\usepackage{graphicx,wrapfig,lipsum}
\usepackage{algpseudocode}
\usepackage{dsfont}
\usepackage{color}
\usepackage{graphicx}
\usepackage{tikz}
\usepackage{amsthm}
\usepackage{hyperref}
\usepackage{caption}
\usepackage{fancyhdr}
\usepackage{subfigure}
\usepackage{microtype}
\usepackage{mathtools}

\usepackage{tikz}
\usetikzlibrary{arrows}

\makeatletter
\newcommand\@erelb@r[1]{%
  \mathrel{\tikz[baseline=-.5ex]\draw[#1] (0,0)--(0.3,0);}
}
\newcommand{\erelbar}[1]{\@erelbar#1}
\def\@erelbar#1#2{%
  \ifcase\numexpr#1*4+#2\relax
    \@erelb@r{-}\or     
    \@erelb@r{->}\or    
    \@erelb@r{-|}\or    
    \@erelb@r{->|}\or   
    \@erelb@r{<-}\or    
    \@erelb@r{<->}\or   
    \@erelb@r{<-|}\or   
    \@erelb@r{<->}\or   
    \@erelb@r{|-}\or    
    \@erelb@r{|->}\or   
    \@erelb@r{|-|}\or   
    \@erelb@r{|<->|}\or 
    \@erelb@r{|<-}\or   
    \@erelb@r{|<->}\or  
    \@erelb@r{|<-|}\or  
    \@erelb@r{|<->|}    
  \else
    \@wrong
  \fi
}
\makeatother

\newfloat{Figure}{tbp}{loa}
\floatname{Figure}{Figure}
\theoremstyle{plain}
\newtheorem{theorem}{Theorem}[section]

\newtheorem{lem}[theorem]{Lemma}
\newtheorem{clm}[theorem]{Claim}

\newtheorem{defn}[theorem]{Definition}

\newtheorem{problem}[theorem]{Problem}
\newtheorem{tech}[theorem]{Technique}
\newtheorem{procedure}[theorem]{Procedure}
\newtheorem{construction}[theorem]{Construction}

\newcommand{\norm}[2]{\left\|#1 \right\|_{#2}}

\usepackage[T1]{fontenc}
\usepackage[latin9]{inputenc}
\usepackage{titlesec}
\usepackage{titling}
\usepackage{verbatim}
\usepackage{float}
\usepackage{mdframed}
\usepackage{amssymb,amsfonts,amsmath}
\usepackage{enumerate}
\usepackage{enumitem}
\usepackage{amsthm}
\usepackage{graphicx,wrapfig,lipsum}
\usepackage{algpseudocode}
\usepackage{dsfont}
\usepackage{color}
\usepackage{graphicx}
\usepackage{tikz}
\usepackage{amsthm}
\usepackage{hyperref}
\usepackage{caption}
\usepackage{fancyhdr}
\usepackage{subfigure}
\usepackage{microtype}
\usepackage{mathtools}

\newcommand{\REF}{ {\color{red} \boxed{REF}} }
\newcommand{\FILL}{ {\color{red} \boxed{FILL}} }

\newcommand{\unisize}{U}
\newcommand{\window}{n}

\title{Tight Bounds for Heavy-Hitters and Moment Estimation in the Sliding Window Model
\thanks{Authors William Swartworth and David Woodruff were supported by a Simons Investigator Award and Office of Naval Research award number N000142112647.}}
\begin{document}
\author{Shiyuan Feng
\\EECS, Peking University\\ \href{mailt:fsyo@stu.pku.edu.cn}{fsyo@stu.pku.edu.cn}
\and
William Swartworth \\Carnegie Mellon University\\ 
\href{mailt:wswartwo@andrew.cmu.edu}{wswartwo@andrew.cmu.edu}
\and David P. Woodruff \\Carnegie Mellon University\\
\href{mailt:dwoodruf@cs.cmu.edu}{dwoodruf@cs.cmu.edu}
}
\date{}

\maketitle

\begin{abstract}
We consider the heavy-hitters and $F_p$ moment estimation problems in the sliding window model.  For $F_p$ moment estimation with $1<p\leq 2$, we show that it is possible to give a $(1\pm \epsilon)$ multiplicative approximation to the $F_p$ moment with $2/3$ probability on any given window of size $n$ using $\tilde{O}(\frac{1}{\epsilon^p}\log^2 n + \frac{1}{\epsilon^2}\log n)$ bits of space.  We complement this result with a lower bound showing that our algorithm gives tight bounds up to factors of $\log\log n$ and $\log\frac{1}{\epsilon}.$  As a consequence of our $F_2$ moment estimation algorithm, we show that the heavy-hitters problem can be solved on an arbitrary window using $O(\frac{1}{\epsilon^2}\log^2 n)$ space which is tight.
\end{abstract}

\section{Introduction}
In some situations, such as monitoring network traffic, data is being received at an enormous rate and it is not always worthwhile to store the entire dataset. Instead we may only be interested in some statistics of the data, which we can hope to estimate from a data structure that uses much less space.  This motivated the development of streaming algorithms, which process a sequence of updates one at a time, while maintaining some small sketch of the underlying data.

In some applications, it makes sense to focus on only the most recent updates.  For instance, when tracking IP addresses over a network \cite{sen2002analyzing, demaine2002frequency}, we might be more interested in tracking the heavy users on a week-to-week basis rather than over longer time periods. Even if we are interested in longer time periods, we may still be interested in obtaining more fine-grained information about how the heavy users change over time.

This type of problem motivated the \textit{sliding window model} for streaming algorithms which has a considerable line of work \cite{datar2002maintaining, 2007, braverman2018nearly, woodruff2022tight}. There are several versions of the sliding window model, but for us our main goal is to respond to queries about the previous $n$ elements of the stream. 

We consider perhaps the two most foundational problems in streaming algorithms: moment estimation and heavy-hitters.  Both of these models have received an enormous amount of attention in the streaming literature, beginning with ~\cite{charikar2002finding} who proposed the CountSketch algorithm for finding heavy items, all the way to \cite{braverman2017bptree} who gave the optimal space bound for insertion-only streams \footnote{This means that items may be added to the stream, but not removed.}, removing one of the $\log$ factors of CountSketch.  While fully tight bounds are known for heavy-hitters (up to a $\log(1/\epsilon)$ factor in insertion streams), in the sliding window model tight bounds are still unknown.  Previous work \cite{braverman2018nearly} gave a $\log^2 n$ lower bound, and a $\log^3 n$ upper bound in the sliding window model.  We close this gap, and show that a $\log^2 n$ dependence suffices for computing heavy-hitters on a particular sliding window.

In the $F_p$ moment estimation problem, the goal is to approximation the $p^{\text{th}}$ moment for the frequency vector of all items inserted into the stream.  Like heavy-hitters, this problem has a long history going back to \cite{alon1996space}.  Since then, a line of work has led to nearly tight bounds both for large moments $p > 2$, which does not admit polylogarithmic sized sketches \cite{bar2004information, indyk2005optimal}, as well as the small moments $1 \leq p \leq 2$ \cite{indyk2006stable, kane2010exact} that we consider here.  While this problem has received considerable attention in the sliding window model \cite{2007, braverman2010effective, woodruff2022tight}, there still remains a gap in the literature.  We aim to close this gap for the problem of estimating the $F_p$ moment on a sliding window.

\subsection{Our models}
\paragraph{The sliding window model.}
There are several slightly different models that one can consider.  By default we consider an insertion-only stream of length $m$ consisting of items $x_1, \ldots, x_m.$  In the sliding window model, there is a window of size $W$ consisting of the $W$ most recent updates to the stream.  A correct sliding window algorithm may be queried for the desired statistic at any time $t,$ and it must produce a correct estimate for the statistic on the stream $x_{m - W + 1}, \ldots, x_m$ of the $W$ most recent updates.  When we refer to the failure probability for such an algorithm, we mean the probability of failure for a single query (for an arbitrary value of $t$).

One could also consider the problem of tracking a stream statistic over the sliding window \textit{at all times}.  This could be accomplished in our framework simply by making the failure probability $1/m$.

One could also consider a more general model where there is no fixed window. Instead at a given query time $t$, the algorithm receives a positive integer $n$ and must estimate the stream statistic on the portion of the stream in $(t-n,t].$  The difference is that $n$ is now not known until runtime.  Our algorithms all apply to this more general model, and our lower bounds hold in the standard sliding window model.  This model may be useful if one wants to observe how the stream has changed up to the present time.  For instance if we were monitoring network traffic, we might be interested in which IP addresses occurred frequently in the past hour, day, week, etc.

\paragraph{$F_p$ moment estimation.}  In the version of the $(\eps, \delta, F_p)$ moment estimation problem that we consider, we receive a series of insertion-only updates $x_1, \ldots, x_m.$  At any point in time $t$ the algorithm may be asked to output an estimate $\hat{F}$ of the $F_p$ moment of the window $W = [t-n+1, t]$ consisting of updates $x_{t-n+1}, \ldots, x_t.$  If $f^{(W)}_i$ is the frequency count of universe item $i$ over the window, then the $F_p$ moment over the window is defined to be $F_P(f^{(W)})=\sum_i (f^{(W)}_i)^p.$  In order to be correct, the algorithm must satisfy $\hat{F} = (1\pm \eps) F_p(f^{(W)})$ with probability at least $1 - \delta.$

\paragraph{Heavy hitters.}  
We say that item $i$ is $(\eps, \ell_p)$-heavy for a frequency vector $f$ if $f_i \geq \eps \norm{f}{p}.$ In the $(\eps, \ell_p)$-heavy hitters problem the goal is to output a collection $\mathcal{H}$ of $O(\eps^{-p})$ universe items such that all $(\eps, \ell_p)$-heavy items for $f^{(W)}$ are contained in $\mathcal{H}$ and such that all elements of $\mathcal{H}$ are at least $c\eps$-heavy where $c > 0$ is a constant.

\subsection{Prior Work}

\cite{2007} introduced the smooth histogram approach for solving problems in the sliding window model. For a function $f$, given adjacent substreams $A, B$, and $C$, an $(\alpha, \beta)$-smooth function demands that if $(1-\beta) f(A \cup B) \leq f(B)$, then $(1-\alpha) f(A \cup B \cup C) \leq f(B \cup C)$ for some parameters $0<\beta \leq \alpha<1$. Throughout the stream, they use the smooth histogram to remove the  unnecessary information and only keep track of the time when $F_p$ differs by $(1 + \epsilon ^p)$. For $p\in (1, 2]$, their work achieves space $O(\epsilon ^ {-(2+p)} \log ^ 3n)$.

\cite{woodruff2022tight} introduced a new tool called a Difference Estimator. The idea is to estimate the difference of $F_p(u+v)$ and $F_p(v)$ with additive error $\epsilon \cdot F_p(v)$, given that $\max(F_p(u), F_p(u+v)-F_p(v)) \le \gamma F(v)$, the dimension of the sketch only needs to be $\tilde{O}(\frac{\gamma^{2/p}}{\epsilon ^ 2})$. They use the idea in \cite{2007} to maintain a constant factor partition on the top level. Then using a binary tree-like structure to give a fine-grained partition, they use a Difference Estimator to exclude the part that is not in the window.

\cite{braverman2017bptree} gives a $\tilde{O}(\epsilon ^ {-2} \log n)$ bits of space algorithm to give strong tracking over $\ell_2$ in insertion-only stream. Also, \cite {blasiok2017continuous} generalize this to $p\in (0, 2)$. They introduce a weak tracking property, which is a key motivation for our \est.

For heavy hitters, \cite{braverman2018nearly} gives a $\tilde{O}(\epsilon^{-p} \log ^ 3n)$ algorithm, with the $\ell_2$ estimation being the bottleneck. If given the $\ell_2$ estimation over the window, then their algorithm can work in space $\tilde{O}(\epsilon ^ {-p} \log ^ 2n)$. This is the result we apply to obtain improved heavy-hitter bounds from our $F_2$ estimator.  They also prove an $\Omega(\epsilon ^ {-p} \log ^ 2n)$ lower bound for the heavy hitter problem in the sliding window model.

\subsection{Our Results}



\begin{table}[htbp]
\centering
\begin{tabular}{|c|c|c|}
\hline
\textbf{Problem (in Sliding Window Model)} & \textbf{Previous Bound} & \textbf{New Bound} \\
\hline
$L_p$-Heavy Hitters, $p \in (0,2]$ & \makecell{ 
$\tilde{O}(\epsilon^{-p}\log^3 n)$\\ $\Omega(\epsilon^{-p}\log^2 n)$ \cite{braverman2018nearly}} & 
$\tilde{\mathcal{O}}(\epsilon^{-p}\log^2 n)$ \\
\hline
$F_p$ Estimation, $p \in (1,2]$ & 
$\tilde{\mathcal{O}}(\epsilon^{-2}\log^3 n)$ \cite{woodruff2022tight} & \makecell{
$\tilde{O}(\epsilon^{-p}\log^2 n + \epsilon^{-2}\log n)$,\\ $\Omega(\epsilon ^ {-p} \log ^ 2n + \epsilon ^ {-2}\log n)$ }\\
\hline
\end{tabular}
\label{tab:results}
\caption{Results in the Sliding Window Model}
\end{table}

\paragraph{$F_p$ estimation.}

We give an algorithm for $F_p$ moment estimation for $1 < p \leq 2$ that achieves optimal bounds in the sliding window model, in terms of the accuracy $\eps$ and the window size $n$.

\begin{theorem}
    There is an algorithm that uses  
    \[
    O\Big((\epsilon ^ {-p} \log ^ 2n + \epsilon ^ {-2}\log n \log ^ 4 \frac 1{\epsilon}) (\log \frac{1}{\epsilon} + (\log \log n) ^ 2)\Big)
    = \tilde{O}(\eps^{-p}\log^2 n + \eps^{-2}\log n)
    \]
    bits of space, and solves the $(\eps, F_p)$ moment estimation problem in the sliding window model. 
\end{theorem}

Note that this result is for approximating the $F_p$ on any \textit{fixed} window with $2/3$ probability.  To track $F_p$ at all times, and guarantee correctness at all times with good probability, one could apply our algorithm $O(\log m)$ times in parallel, and take the median estimate at each time.

\paragraph{Heavy hitters.}
As a consequence of our $F_2$ approximation scheme, we obtain an optimal heavy-hitters algorithm by combining with the results of \cite{braverman2018nearly}.

\begin{theorem}\label{Heavy Hitter}
Given $\epsilon>0$ and $0 < p \leq 2$, there exists an algorithm in the sliding window model that, with probability at least $\frac{2}{3}$, outputs all indices $i \in[U]$ for which frequency $f_i \geq \epsilon \ell_p ^ W$, and reports no indices $i \in[U]$ for which $f_i \leq \frac{\epsilon}{12} \ell_p ^ W$. The algorithm has space complexity $O\left(\eps^{-p} \log ^2 n\left(\log \frac{1}{\epsilon}+\log \log n\right)\right)$.
\end{theorem}
An earlier version of \cite{braverman2018nearly} claimed this bound for the stronger problem of tracking the heavy-hitters on the window at all times.  Unfortunately this original bound was incorrect, and their algorithm actually gives a $\log^3 n$ dependence, even in our setting where we are interested in obtaining good success probability for a single window.

\paragraph{Lower bounds.}  We also show that our $F_p$ estimation algorithms are  tight in the sliding window model up to $\log\frac{1}{\eps} \log\log n$ factors.  Specifically, we show

\begin{theorem}
Fix $p \in (1,2]$. Suppose that $\mathcal{A}$ is a streaming algorithm that operates on a window of size $\window$, and outputs a $(1\pm \eps)$ approximation of $F_p$ on any fixed window with $9/10$ probability. For $\window \geq \unisize,$ $\mathcal{A}$ must use at least \[\Omega\left(\eps^{-p}\log^2(\eps \unisize) + \eps^{-2}\log(\eps^{1/p} \unisize)\right)\] space.
\end{theorem}

The first term in the lower bound is our contribution.  The second term applies to general insertion-only streams, and follows from a recent work of \cite{braverman2024optimality}. 
While the lower bound only applies to $p>1$, our proof also applies to $p=1$ if the stream is allowed to contain empty insertions that do not affect the frequency counts, but that move the window. (Without empty insertions $p=1$ is trivial, since $F_1$ is always the window size.)

\subsection{Notation}

We use $m$ to denote the length of the stream, $U$ be the universe size, and $u_1, \dots, u_m \in [U]$ to denote the elements in the stream. For sliding window, we use $n$ to denote the window size. We use $\text{poly}(n)$ to denote a constant degree polynomial in $n$. We assume $m \le \text{poly}(n), U \le \text{poly}(n)$. 

For interval $[l, r] \subseteq [1, m]$, we use $x ^ {(l,r)}$ to denote the frequency vector of elements $u_l, u_{l+1}, \dots, u_r$. For $p \in (0, 2]$, we use $\ell_p ^ {(l,r)}$ to denote the $\ell_p$ norm of the frequency vector $x ^ {(l,r)}$, more specifically, $\|x ^ {(l, r)}\|_p$, and $F_p ^ {(l,r)}$ to denote the $F_p$ moment of frequency vector $x ^ {(l,r)}$, more specifically, $\|x ^ {(l, r)}\|_p ^ p$. For window $W = [l, r]$, we also use $\ell_p ^ W$ or $F_p ^ W$ to denote the $\ell_p$ norm or $F_p$ frequency moment on window $W$. We use $a = (1 \pm \epsilon) b$ to denote $a \in [(1 - \epsilon)b, (1 + \epsilon)b]$.

\subsection{Technical Overview}

\paragraph{$F_p$ estimation.} 

To motivate our approach, consider approximating $F_2$ to within a constant factor. We first recall the framework introduced by \cite{braverman2007smooth} for the sliding window model.   The rough idea is to start a new $F_p$ estimation sketch at each point in the stream.  This clearly gives a sliding window estimate, but uses far too much space.  However, the observation is that if sketch $i$ and sketch $i+2$ give similar estimates for $F_p$ at a given point in the stream, then they will continue to do so for the remainder of the stream.  The estimate given by sketch $i+1$ is sandwiched between the estimates from the neighboring sketches, so sketch $i+1$ is redundant, and we may safely delete it.  After deleting all redundant sketches, we are left with $O(\log n)$ sketches (for constant $\eps$) with geometrically increasing $F_p$ estimates.

Each of these $F_2$ sketches uses $O(\log n \log\frac{1}{\delta})$ bits of space to estimate $F_2$ with probability $1-\delta$, so this approach requires $O(\log^2 n \log \frac{1}{\delta})$ bits.  Since we only store $\log n$ sketches at a time, it might seem that we can take $\delta = O(\frac{1}{\log n}).$  This is a subtle point and has been a source of error in prior work.

In fact we cannot take $\delta = O(\frac{1}{\log n})$ if we are only guaranteed that our sketches correctly estimate $F_2$ with probability at least $1-\delta$.  To see why, suppose that each sketch we create is correct with probability at least $1-\delta$ and is otherwise allowed to be adversarially incorrect with probability $\delta.$  With $\delta = O(1/\log n),$ it is likely that two consecutive sketches will eventually both be incorrect.  In that case, they can conspire to in turn sandwich and wipe out all sketches in front of them, from which we will not be able to recover.  While extreme, this example illustrates that an $F_p$ estimation sketch with a failure probability guarantee is not sufficient for us.  This is why prior approaches for estimating the  $\ell_p$ norm or $F_p$ moment in sliding windows require $\Omega(\log ^ 3n)$ bits of space, with one of the logarithmic factors arising from the need for sketches with high success rates. These high-success-rate sketches necessitate union bounds over all sketches, even though most of them are not kept by the algorithm when we make a query. 

Of course we should not expect sketches to fail this way in practice.  Could there be a somewhat stronger guarantee that rules out this type of adversarial behavior?  In this work, we overcome this limitation by introducing a new tool: the \est. This estimator enables a refined analysis that avoids the union-bound overhead, allowing us to achieve the first space complexity which breaks the  $O(\log^ 3 n)$ bit barrier. 

The intuition of \est is for a window $[l, r]$, we do not need our estimator to be right (meaning give an $\epsilon$ approximation) on any sub-interval of this window, but rather only an additive error: $\epsilon$ fraction of the $\ell_p$ norm of the window suffices for maintaining the timestamps. Thus we can avoid the union bound for all sub-intervals to be right. The intuition is similar to the weak tracking in \cite{blasiok2017continuous}.

In Section $3$, we show that \est with the simplest algorithm in \cite{2007} suffices to give an $\epsilon$ approximation to the $\ell_p$ norm of the sliding window. However, this simple algorithm uses $\tilde{O}(\epsilon ^ {-3p} \log ^ 2n)$ bits of space, which requires a high dependence on $\epsilon$. 

In Section $4$, we introduce \dif \cite{woodruff2022tight} to improve the $\epsilon$ dependence. We first use our simple algorithm to give a constant factor partition on the top level. We then use a binary tree-like structure to make a more detailed partition between the top level, and use \dif estimator to exclude the contribution that is not in the window. The algorithm in \cite{woodruff2022tight} requires $\tilde{O}(\epsilon ^ {-2} \log ^ 3n)$ space. The extra $\log n$ factor also comes from the high success probability requirement. Using the \est and the structure of their algorithm, we can obtain an algorithm using $O(\epsilon ^ {-2} \log ^ 2n)$ bits of memory. We further discover that the bottleneck of the algorithm comes from the \dif, which can be stored by rounded values using the rounding procedure in \cite{rounding}. We thus improve the space to $\tilde{O}(\epsilon ^ {-p} \log ^ 2n + \epsilon ^ {-2} \log n)$ bits, which gives a separation between $p = 2$ and $p < 2$. We further prove a lower bound on $F_p$ moment estimation over sliding windows, see below, showing that our algorithm is nearly optimal.

\paragraph{Heavy Hitters.}
\cite{braverman2018nearly} showed that if one has an $\ell_2$ estimation algorithm over a sliding window, then one can find the $\ell_2$ heavy hitters with $\tilde{O}(\epsilon ^ {-2} \log ^ 2n)$ additional bits of space. In \cite{heavyhitter}, they show that an $\epsilon ^ {p/2}$ - heavy hitter algorithm for $\ell_2$ also finds $\epsilon$- heavy hitters for $\ell_p$. Thus, using our simplest algorithm to provide a constant $\ell_2$ estimation over the window will give us an algorithm that uses $\tilde{O}(\epsilon ^ {-p} \log ^ 2n)$ bits of space, and returns all the $\epsilon$-$\ell_p$ heavy hitters. This matches the lower bound $\Omega(\epsilon ^ {-p}\log ^ 2n)$ for heavy hitters in \cite{braverman2018nearly}.

\paragraph{Lower Bounds.}
For our lower bound for $F_p$ estimation, we construct a stream that consists of roughly $\log(n)$ blocks, where each block has half the $F_p$ value of the prior block.  Within each block, we place $\eps^{-p}$ items that are each $(\eps, \ell_p)$-heavy for the block, and arrange them so that each of these heavy items occurs in a contiguous ``mini-block".  Overall, there are $\frac{1}{\eps^p} \log(n)$ of these ``mini-blocks".

The main observation is that a correct streaming algorithm for $F_p$-moment estimation in the sliding window model must remember the location of each mini-block. Doing so requires $\log n$ bits of space per mini-block as long as there are at least, say, $n^{0.1}$ disjoint positions that each block can occur in, resulting in our $\Omega(\frac{1}{\eps^p}\log^2 n)$ lower bound.

To see why our algorithm must remember each location, consider one of the heavy items $x$ in our construction, and suppose we want to decide whether the mini-block for $x$ occurs entirely before or entirely after stream index $i.$  We first shift the sliding window so that it begins at index $i$.  We would like to decide whether or not $x$ occurs in the resulting window.  By the geometric decay of the blocks, $x$ remains heavy for the entire window. Then to decide if $x$ remains in the window we append $Q$ copies of $x$ to the window, where $Q$ is a rough estimate of the $F_p$ over the window.  If $x$ does not occur in the window, then appending the $x$'s increases the $F_p$ by $Q^p.$  On the other hand, if $x$ does occur roughly $\eps Q$ times in the window then the $F_p$ increases by $(\eps Q + Q)^p - (\eps Q)^p = \Theta((1 + \eps)Q^p).$ Since $Q^p$ is approximately the $F_p$ moment of the window, a $\Theta(\eps)$ heavy-hitters algorithm can distinguish between these two cases.  A minor issue is that as stated, we would need to append the $x$'s without shifting the window past the index $i$. 
 However we can easily fix this by appending the $x$'s before performing appending singletons to accomplish the shift.  To make the above intuition precise we give a reduction from the IndexGreater communication game in Section~\ref{sec:lower_bounds}.  

\section{Preliminary Results}

In this section, we will introduce some basic definitions and lemmas that form the foundation of our \estt. We will also introduce the previous results for heavy hitters. We first require the following definition for $p$-stable distribution. 

\begin{defn}( $p$-stable distribution).\cite{zolotarev1986one} For $0<p \leq 2$, there exists a probability distribution $\mathcal{D}_p$ called the $p$-stable distribution so that for any positive integer $n$ with $Z_1, \ldots, Z_n \sim \mathcal{D}_p$ and vector $x \in \mathbb{R}^n$, then $\sum_{i=1}^n Z_i x_i \sim\|x\|_p Z$ for $Z \sim \mathcal{D}_p$.
\end{defn}

The probability density function $f_X$ of a $p$-stable random variable $X$ satisfies $f_X(x)=\Theta\left(\frac{1}{1+|x|^{1+p}}\right)$ for $p<2$, while the normal distribution corresponds to $p=2$. Moreover, \cite{Nolan2001StableDM} details standard methods for generating $p$-stable random variables by taking $\theta$ uniformly at random from the interval $\left[-\frac{\pi}{2}, \frac{\pi}{2}\right], r$ uniformly at random from the interval $[0,1]$, and setting

$$
X=f(r, \theta)=\frac{\sin (p \theta)}{\cos ^{1 / p}(\theta)} \cdot\left(\frac{\cos (\theta(1-p))}{\log \frac{1}{r}}\right)^{\frac{1}{p}-1}
$$

These $p$-stable random variables are crucial to obtaining our \estt.

\begin{defn}\label{Strong Estimator}(\estt)
    Let $\epsilon_1 \leq \epsilon_2 \in (0, 1)$,  $ p \in (0, 2]$ and let $u_1, u_2, \dots, u_m$ be the sequence of stream elements. Let $x ^ {(i, j)}$ denote the frequency vector for elements $u_i, \dots, u_j.$
 We say that $f_p$ has the $(\eps_1, \eps_2, \delta)$  \est property on the window $[l,r]$ if with probability $1-\delta$, we have the bound $f_p(x ^ {(a,b)}) = (1 \pm \epsilon_1) \ell_p ^ {(a, b)} \pm  \epsilon_2 \ell_p ^ {(l,r)}$ for all sub-windows $[a,b]\subseteq [l,r]$ simultaneously. We say that $f_p$ has the $(\eps_1, \eps_2, \delta)$ \est property if it has the $(\eps_1, \eps_2, \delta)$ \est property on all windows $[l,r].$ 

    For simplicity, we will use $f_p ^ {(a, b)}$ to denote the estimate $f_p(x ^ {(a, b)})$. 
\end{defn}

To support the construction and analysis of such an estimator, we rely on several probabilistic tools. Lemma \ref{cher} gives a Chernoff-type concentration bound for $k$-wise independent random variables. Lemma \ref{prefix bound} provides tail bounds on the supremum of inner products between a $p$-stable random vector and a sequence with an increasing frequency vector. 

\begin{lem}\label{cher}\cite{bellare1994randomness} Let $X_1, \ldots X_n \in\{0,1\}$ be a sequence of $k$-wise independent random variables, and let $\mu=\sum \mathbb{E} X_i$. Then

$$
\forall \lambda>0, \mathbb{P}\left(\sum X_i \geq(1+\lambda) \mu\right) \leq \exp \left(-\Omega\left(\min \left\{\lambda, \lambda^2\right\} \mu\right)\right)+\exp (-\Omega(k))
$$
    
\end{lem}

\begin{lem}\label{prefix bound} \cite{blasiok2017continuous} Let $x^{(1)}, x^{(2)}, \ldots x^{(m)} \in \mathbb{R}^n$ satisfy $0 \preceq x^{(1)} \preceq x^{(2)} \preceq \ldots \preceq x^{(m)}$. Let $Z \in \mathbb{R}^n$ have $k$-wise independent entries marginally distributed according to $\mathcal{D}_p$. Then for some $C_p$ depending only on $p$,

$$
\mathbb{P}\left(\sup _{k \leq m}\left|\left\langle Z, x^{(k)}\right\rangle\right| \geq \lambda\left\|x^{(m)}\right\|_p\right) \leq C_p\left(\frac{1}{\lambda^{2 p /(2+p)}}+k^{-1 / p}\right)
$$

\end{lem}

Finally, we include two additional lemmas that are essential for our heavy hitter detection algorithm in the sliding window model.

Lemma~\ref{2}, due to Braverman et al.~\cite{braverman2018nearly}, shows that a constant-factor approximation to the $\ell_2$ norm within a window suffices to recover all $\epsilon$-heavy hitters under the $\ell_2$ norm, with provable accuracy guarantees and near-optimal space complexity.

\begin{lem}\label{2}\cite{braverman2018nearly}
    For a fixed window $W$, if given a $\frac 12$ approximation of $\ell_2 ^ {W}$, then there is an algorithm that outputs all the $\epsilon$ heavy hitters for $\ell_2$ within the window, and can guarantee that all the output elements have frequency $\ge \frac{\epsilon}{12} \ell_2 ^ W$. This algorithm uses $O(\epsilon ^ {-2} \log ^ 2n (\log \frac 1{\epsilon}+\log \log n))$ and has success probability $\ge \frac 23$.
\end{lem}

To extend this result to the $\ell_p$ setting for any $p \in (0,2]$, we use Lemma~\ref{heavyhitter}, which establishes that any algorithm capable of recovering $\epsilon^{p/2}$-heavy hitters under $\ell_2$ (with a suitable tail guarantee) also successfully identifies all $\epsilon$-heavy hitters under the $\ell_p$ norm.

\begin{lem}\cite{heavyhitter} \label{heavyhitter}For any $p \in(0,2]$, any algorithm that returns the $\epsilon^{p / 2}$-heavy hitters for $\ell_2$ satisfying the tail guarantee also finds the $\epsilon$-heavy hitters for $\ell_p$.
\end{lem}



        




    

\section{Warmup: Constant-factor $F_p$ approximation}

In this section, we will introduce our result for \estt. This estimator enables us to design a constant-factor approximation algorithm for $F_p$, and in combination with previous results, it leads to a nearly optimal algorithm for identifying $\epsilon$-heavy hitters over sliding windows.

\subsection{Constructing a Strong Estimator}

First, we will show how to construct a space efficient \est defined in Definition \ref{Strong Estimator}. 

\begin{construction}\label{construct est}
    
Let $d$, $r$ and $s$ be parameters to be specified later, the  \est takes $\Pi \in \mathbb{R}^{d \times U}$ to be a random matrix with entries drawn according to $\mathcal{D}_p$, and such that the rows are r-wise independent, and all entries within a row are s-wise independent. For frequency vector $x ^ {(a, b)}$, the estimate $f_p(x ^ {(a, b)})$ is given by  $\med(|(\Pi x ^ {(a, b)})_1|,|(\Pi x ^ {(a, b)})_2|, \dots, |(\Pi x ^ {(a, b)})_d|)$. 
        
\end{construction}

The following lemma shows that, by appropriately setting the parameters $d$, $r$, and $s$, the above construction yields a valid \estt.

\begin{lem} \label{est}

Let $\epsilon_1, \epsilon_2 \in (0, 1)$ with $\epsilon_2 \le \epsilon_1$, $\delta \in (0, 1)$, and $p \in [1, 2]$. Given a stream of elements $u_1, u_2, \dots, u_m$, Construction~\ref{construct est} yields an $(\epsilon_1, \epsilon_2, \delta)$ \est by setting the parameters as follows:
\[
d = \Theta\left(\epsilon_1^{-2} \left( \log \tfrac{1}{\epsilon_2} + \log \tfrac{1}{\delta} \right) \right), \quad
r = \Theta\left( \log \tfrac{1}{\epsilon_2} + \log \tfrac{1}{\delta} \right), \quad
s = \Theta\left( \epsilon_1^{-p} \right).
\]
\end{lem}

To show the space efficiency, we also analyze the space complexity of storing the matrix $\Pi$ used in the construction.

\begin{lem}\label{space}
The memory required to store the matrix $\Pi$ for an $(\epsilon_1, \epsilon_2, \delta)$ \est is 
\[
O\left(\epsilon_1^{-p} \left( \log \tfrac{1}{\epsilon_2} + \log \tfrac{1}{\delta} \right) \log (mU)\right) \text{ bits}.
\]
\end{lem}

Last, we will show an additional property of \estt, which is crucial for solving $F_p$ estimation over sliding windows.  

\begin{lem}\label{add} (Additional Property for the \estt)
    The $(\eps_1, \eps_2, \delta)$ \est $f_p$ has the following property:

For a sequence of elements in the stream $u_1, u_2, \dots, u_m$ and a window $[l, r]$, with probability $\ge 1 - \delta \cdot \log r$, for any interval $[a, b]$ that $1 \le a < l, l \le b \le r$, $f_p(x ^ {(a, b)}) = (1 \pm (\epsilon_1 +2 \epsilon_2)) \ell_p ^ {(a, b)} \pm 2\epsilon_2 \ell_p ^ {(l,r)} $.  
    
\end{lem}

\subsection{Constant-factor $F_p$ approximation }

First, we will introduce the algorithm given in \cite{2007}: it maintains $O(\epsilon^{-p} \log n)$ timestamps to track the times at which $\ell_p$ changes by a factor of $(1 + \epsilon^p)$. This algorithm requires $O(\log^3 n)$ bits of space, where one $\log n$ factor accounts for the number of sketches, another $\log n$ arises from the need for high success probability in order to apply a union bound over all sketches, and the final $\log n$ comes from the fact that each entry in the sketch requires $\log n$ bits.

Our key improvement over their algorithm is that, by using the \est, we no longer require the sketches to have high success probability. The \est provides guarantees over all sub-intervals directly, eliminating the need for a union bound.

We present an $\epsilon$-approximation algorithm for $\ell_p$ over sliding windows that uses $\tilde{O}(\epsilon^{-3p} \log^2 n)$ space. This is the first algorithm achieving $o(\log^3 n)$ space for this problem. As $\ell_p$ estimation and $F_p$ estimation are fundamentally the same, we do not distinguish between the two here. 

The following theorem formalizes the guarantees of our approach:

\begin{theorem}\label{1}
    For any fixed window $W$, with probability $\ge 1 - \delta$, Algorithm \ref{alg:reg1} will give an $\epsilon$ approximation of $\ell_p ^ {W}$, and  use $O(\epsilon ^ {-3p} \log ^ 2n (\log \frac 1{\epsilon} + \log \frac 1{ \delta}+  \log \log n))$ bits of space.
\end{theorem}



We now apply our algorithm to the heavy hitter problem in sliding windows. Combining our $\frac{1}{2}$-approximation for $\ell_2$ with known results for heavy hitters (Lemma \ref{2} and Lemma \ref{heavyhitter}), we obtain the following result:

\begin{theorem}\label{Heavy Hitter}
(Main Theorem on Heavy Hitters)
    Given $\epsilon>0$ and $0<p \leq 2$, there exists an algorithm in the sliding window model that, with probability at least $\frac{2}{3}$, outputs all indices $i \in[U]$ for which frequency $f_i \geq \epsilon \ell_p ^ W$, and reports no indices $i \in[U]$ for which $f_i \leq \frac{\epsilon}{12} \ell_p ^ W$. The algorithm has space complexity (in bits) $O\left(\frac{1}{\epsilon^p} \log ^2 n\left(\log \frac{1}{\epsilon}+\log \log n\right)\right)$.
\end{theorem}

\begin{proof}
We use Algorithm \ref{alg:reg1} to obtain a $\frac 12$ approximation over $\ell_p ^ W$, the rest of the proof follows by Lemma \ref{2} and Lemma \ref{heavyhitter}.

\end{proof}

\begin{algorithm*}[t!]
\caption{\textbf{\textsf{Smooth Histogram with \est}}} \label{alg:reg1}

{\bf Input:} Stream $a_1, a_2, \dots \in [U] $, window length $n$, approximate ratio $\epsilon \in (0, 1)$ and error rate $\delta \in (0,1)$, parameter $p$.

{\bf Output:} At each time $i$, output a $(1\pm \epsilon)$ approximation to $\ell_p$  in window $[i - n, i]$.

{\bf Initialization:} \begin{enumerate}
    \item Generate the matrix $\Pi \in \R^{d \times U}$ for an $(\frac{\epsilon ^ p}{64}, \frac{\epsilon ^ p}{128}, \delta'=\frac {\delta}{O( \log m )})$ \estt, where $d = O(\epsilon ^ {-2p} \log \frac 1{\epsilon \delta'}))$. 

    \item Random matrix $\pi \in \R^{d'\times U}$ for a \indd, where $d' = O(\epsilon ^ {-2} \log \frac 1{\delta})$.  
\end{enumerate}

 Let $f ^ {(D)}$ be the estimated value of $\|x ^ {(D)}\|_p$ from \est on data stream $D$, and $g^ {(D)}$ be the estimated value of $\|x ^ {(D)}\|_p$ from \indd.

\noindent\makebox[\linewidth]{\rule{\paperwidth-10cm}{0.4pt}}

{\bf During the Stream:}
\begin{algorithmic}[1]
\State $ s \from$ 0, $i \from 1$.
\For{each update $a_i$} 

    \State $s = s + 1$. 
    \State let $t_s = i$, start a copy of \est using matrix $\Pi$ and \ind using matrix $\pi$ from $t_s$. 
    \For {$j$ in $0, 1, \dots, s$}
        \While {$f ^ {(t_{j + 2},i)} \ge (1 - \frac{\epsilon ^ p}{32})f ^ {(t_j, i)}$} 
            \State Delete $t_{j + 1}$, update indices, $s = s -1$
        \EndWhile
    \EndFor
    \State If $t_2 < i - n$, delete $t_1$, update indices, $s = s -1$. 
    \State Output $g ^ {(t_2 , i)}$ as the estimation of $\ell_p ^ W$.
\EndFor

\end{algorithmic}
\end{algorithm*}



    





Define $t_k(j)$ to be the $k$-th timestamps when the stream goes to $j$. The following lemma shows that the $\ell_p$ norm does not decreases significantly between two consecutive timestamps.
\begin{lem}\label{Eps Approx}
    Let $p \in [1, 2]$, for fixed window $W = [i - n, i]$, with probability $\ge 1 - \delta$, either $t_1(i) + 1 = t_2(i)$, or $\ell_p ^ {(t_2(i), i)} \ge (1 - \frac {\epsilon}2) \ell_p ^ {(t_1(i), i)}$. 
\end{lem}

Now we prove theorem \ref{1}. 
\begin{proof}(Proof of Theorem \ref{1})

Our \ind will give a $\frac {\epsilon}{4}$ approximation with probability $\ge 1 - \delta$. By a union bound over the event of Lemma \ref{Eps Approx} and the condition that \ind is right, by adjusting the constant factor for $\delta$, we obtain an $\epsilon$ approximation of $\ell_p ^ W$ with probability $\ge 1 - \delta$.

 For the space, as the value $f$ is bounded by $\text{poly}(n)$, the number of timestamps will be bounded by $O(\epsilon ^ {-p} \log n)$. 
 
 For each timestamp, we keep a vector with dimension $d = O(\epsilon ^ {-2p} \log \frac {\log m}{\epsilon \delta})$ for the \est, and a vector with dimension $d = O(\epsilon ^ {-2} \log \frac 1{\delta})$ for \ind. So the total space for the sketches is $O(\epsilon ^ {-3p} \log ^ 2n (\log \frac{1}{\epsilon} + \log \log n + \log \frac 1{\delta}))$, as we have assumed $m\le \text{poly}(n)$. The space to store the matrix is of lower order by Lemma \ref{space}.

\end{proof}

\section{Improving the $\eps$ dependence}
\subsection{Preliminary Results} 

Here we use $F(u)$ to denote $\|u\|_p ^ p$ for frequency vector $u$. 

First, we introduce the definition and construction of \diff, which is essential in our algorithm that gives tight bound on $F_p$ estimation over sliding window. 

\begin{defn}\label{dif} (Difference Estimator) (Definition 8.2, \cite{woodruff2022tight})

Given a stream $\mathcal{S}$ and fixed times $t_1$, $t_2$, and $t_3$, let frequency vectors $u$ and $v$ be induced by the updates of $\mathcal{S}$ between times $\left[t_1, t_2\right)$ and $\left[t_2, t_3\right)$. Given an accuracy parameter $\varepsilon>0$ and a failure probability $\delta \in(0,1)$, a streaming algorithm $\mathcal{C}\left(t_1, t_2, t, \gamma, \varepsilon, \delta\right)$ is a $(\gamma, \varepsilon, \delta)$-suffix difference estimator for a function $F$ if, with probability at least $1-\delta$, it outputs an additive $\varepsilon \cdot F\left(v+w_t\right)$ approximation to $F\left(u+v+w_t\right)-F\left(v+w_t\right)$ for all frequency vectors $w_t$ induced by $\left[t_3, t\right)$ for times $t>t_3$, given $\min (F(u), F(u+v)-F(v)) \leq \gamma \cdot F(v)$ for a ratio parameter $\gamma \in(0,1]$.
    
\end{defn}

\begin{lem}\label{dif lemma} (Construction of Difference Estimator) (Lemma 4.9, Lemma 3.6 \cite{woodruff2022tight})

Let the dimension $d = \mathcal{O}\left(\frac{\gamma^{2 / p}}{\varepsilon^2}\left(\log \frac{1}{\varepsilon}+\log \frac{1}{\delta}\right)\right)$.

For $0<p<2$, it suffices to use a sketching matrix $A \in \mathbb{R}^{d \times U}$ of i.i.d. entries drawn from the $p$-stable distribution $\mathcal{D}_p$, with dimension $d$ to obtain a $(\gamma, \varepsilon, \delta)$-difference estimator for $F_p$. To store such a matrix $A$ requires $O(d\log n(\log \log n) ^ 2)$  space. 

Moreover, the estimation of $F(u + v)- F(v)$ is given by: 
\begin{enumerate}
    \item For a parameter $q=3$, let each $z_i=\prod_{j=q(i-1)+1}^{q i}\left(A u+A v\right)_j^{p / q}$ and $z_i^{\prime}=$ $\prod_{j=q(i-1)+1}^{q i}(A v)_j^{p / q}$.
    \item Output the arithmetic mean of $\left(z_1-z_1^{\prime}\right),\left(z_2-z_2^{\prime}\right), \ldots,\left(z_{d / q}-z_{d / q}^{\prime}\right)$.
\end{enumerate}

For $p = 2$, it suffices to use a sketching matrix $M \in \mathbb{R} ^ {d\times U}$ that each entry $M_{i,j}$ of $M$ is a 4-wise independent random sign scaled by $\frac 1{\sqrt d}$ to obtain a $(\gamma, \epsilon, \delta)$-difference estimator for $F_2$. To store such a matrix $M$ requires $O(d \log n)$ space.

The estimation of $F(u+v) - F(v)$ is given by $\|M(u+v)\|_2 ^ 2 -\|Mv\|_2^2$.
    
\end{lem}

Next, we introduce the rounding procedure described in \cite{rounding}. Although this procedure is originally designed for a distributed setting, the process of merging two sketches can be interpreted as two nodes transmitting their respective sketches to a parent node, which then performs a new rounding operation based on the information received from its children.

\begin{defn}\cite{rounding}(Rounding random variable)  \label{rounding variable}
    For any real value $r \in \mathbb{R}$, let $i_r \in \mathbb{Z}$ and $\alpha_i \in\{1,-1\}$ be such that $(1+\gamma)^{i_r} \leq \alpha_i r \leq(1+\gamma)^{i_r+1}$. Now fix $p_r$ such that:

$$
\alpha_i r=p_r(1+\gamma)^{i_r+1}+\left(1-p_r\right)(1+\gamma)^{i_r}
$$

We then define the rounding random variable $\Gamma_{\gamma}(r)$ by

$$
\Gamma_{\gamma}(r)= \begin{cases}0 & \text { if } r=0 \\ \alpha_i(1+\gamma)^{i_r+1} & \text { with probability } p_r \\ \alpha_i(1+\gamma)^{i_r} & \text { with probability } 1-p_r\end{cases}
$$
    
\end{defn}

\begin{procedure}(Recursive Randomized Rounding) \label{distributed rounding procedure}

1. Choose random vector $Z \in \mathbb{R}^n$ using shared randomness.

2. Receive rounded sketches $r_{j_1}, r_{j_2}, \ldots, r_{j_{t_j}} \in \mathbb{R}$ from the $t_j$ children of node $j$ in the prior layer (if any such children exist).

3. Compute $x_j=\left\langle X_j, Z\right\rangle+r_{j_1}+r_{j_2}+\cdots+r_{j_t} \in \mathbb{R}$.

4. Compute $r_j=\Gamma\left(x_j\right)$. If player $j \neq \mathcal{C}$, then send $r_j$ it to the parent node of $j$ in $T$. If $j=\mathcal{C}$, then output $r_j$ as the approximation to $\langle Z, \mathcal{X}\rangle$.
\end{procedure}

\begin{lem}\label{rounding} (Lemma 2, \cite{rounding})

    
    

    Let $U$ be the universe space, and $m$ be the total point in the graph, $d$ be the diameter of this graph, $\epsilon, \delta \in (0, 1)$. Fix $p \in[1,2]$, and let $Z=\left(Z_1, Z_2, \ldots, Z_U\right) \sim D_p^U$. Then the above procedure when run on $\gamma=(\epsilon \delta /(d \log (U m)))^C$ for a sufficiently large constant $C$, produces an estimate $r_C$ of $\langle Z, \mathcal{X}\rangle$, held at the center vertex $\mathcal{C}$, such that $\mathbb{E}\left[r_{\mathcal{C}}\right]=\langle Z, \mathcal{X}\rangle$. Moreover, over the randomness used to draw $Z$, with probability $1-\delta$ for $p<2$, and with probability $1-e^{-1 / \delta}$ for Gaussian $Z$, we have $\mathbb{E}\left[\left(r_{\mathcal{C}}-\langle Z, \mathcal{X}\rangle\right)^2\right] \leq(\epsilon / \delta)^2\|\mathcal{X}\|_p$. Thus, with probability at least $1-O(\delta)$, we have

$$
\left|r_{\mathcal{C}}-\langle Z, \mathcal{X}\rangle\right| \leq \epsilon\|\mathcal{X}\|_p
$$

Moreover, if $Z=\left(Z_1, Z_2, \ldots, Z_n\right) \in \mathbb{R}^n$ where each $Z_i \in\{1,-1\}$ is a 4-wise independent Rademacher variable, then the above bound holds with $p=2$ (and with probability $1-\delta$ ).
\end{lem}

Now, we give our own rounding procedure.

\begin{procedure} \label{rounding procedure}

Suppose we want to update the rounded sketch of window $[l, r]$. The sketch has a matrix $\Pi \in D_p ^ {d\times U}$ for $p \in (0, 2)$. For $p = 2$, the sketch has a matrix  $\Pi \in \{1, -1\} ^ {d\times U}$, with each entry being a Rademacher variable and each row being 4-wise independent. The precise sketch for window $[l, r]$ is $(|\Pi x ^ {(l,r)}|_1, \dots, |\Pi x ^ {(l,r)}|_d) \in \mathbb{R} ^ d$. Let $\mathcal{P}(\Pi x ^ {(l,r)}) \in \mathbb{R} ^ d$ be the sketch vector after the rounding procedure. 

There are 2 cases: First is given a non-rounded sketch, and then rounding the sketch. In this case, for each $i\in [d]$, $\mathcal{P}(\Pi x ^ {(l,r)})_i = \Gamma_{\gamma} ((\Pi x ^ {(l,r)})_i)$.

The second case is given the rounded sketch of 2 sub-intervals $[l, k]$ and $[k+1, r]$. For this case, $\mathcal{P}(\Pi x ^ {(l,r)})_i = \Gamma_{\gamma}(\mathcal{P}(\Pi x ^ {(l,k)})_i + \mathcal{P}(\Pi x ^ {(k+1,r)})_i)$.
\end{procedure}

\begin{lem}\label{rounding lemma}
    Let $[l, r]$ be a window, $\Pi \in \mathbb{R} ^ {d\times U}$ be the sketching matrix (meaning each entry is sampled from $\mathcal{D}_p$ if $0 < p < 2$, or is a Rademacher variable if $p = 2$), and we are running procedure $\mathcal{P}$ described in Procedure \ref{rounding procedure} with parameter $\gamma = (\frac{\epsilon \delta d}{m\log (Um)}) ^ C$ to get a rounded vector for sketch $\Pi x ^ {(l,r)}$. Then for all $i\in [d]$, we have $\Big|\mathcal{P}(\Pi x ^ {(l,r)})_i - (\Pi x ^ {(l,r)})_i\Big| \le \epsilon \|x ^ {(l, r)}\|_p$. 
\end{lem}

\subsection{Algorithm Overview} 



The structure of our algorithm follows \cite{woodruff2022tight}. Let the current time be $r$ and query window be $W = [r-n+1,r]$. Our algorithm first uses Algorithm \ref{alg:reg1} to maintain a constant factor partition on the top level. We will have $O(\log n)$ timestamps on the top level. Let them be $t_1< t_2< \dots< t_s$. 

For each of the $O(\log n)$ pieces of the stream, we will partition the substream $[t_i, t_{i + 1}]$ into more detailed timestamps. We will have a binary tree-like structure, with $\beta = O(\log \epsilon ^ {-p})$ at the top level, and each level $j$ will keep roughly $O(2 ^ j)$ timestamps. Define the $k$-th timestamp on level $j$ between top level $i$ to be $t_{i, j, k}$. 

We will use \dif to estimate the difference of $F_p$ between 2 nearby timestamps and the current time $r$, to exclude the contribution from $t_1$ to $r-n$. As the \dif is maintained by a binary tree-like structure, we only need to query the \dif a constant number of times on each level. So if each \dif gives an $\eta$ additive error of $F_p ^ W$, then our final approximation will have an $O(\eta \cdot \beta)$ additive error of $F_p ^ W$. So a value $\eta = O(\frac {\epsilon}{\log \frac 1{\epsilon}})$ suffices and the \dif uses $O(\epsilon ^ {-2} \log ^ 2n \text{poly}(\log \frac 1{\epsilon}))$ space in total. 

To maintain the binary tree-like structure for the \diff, the algorithm by \cite{woodruff2022tight} uses a $p$-stable sketch with high success probability, and thus by union bound, these sketch always can rule out unnecessary timestamps and retain only the essential ones. 

Our algorithm relies on the analysis of \estt, and no such union bound will be used (thus saving a factor of $\log n$), so it will require a more careful analysis. In the following content, we introduce our new techniques.

\begin{algorithm}[H]
\caption{\textbf{\textsf{Moment Estimation in the Sliding Window with Optimal Space}}}
\label{alg2}
\textbf{Input:} Stream $u_1, u_2, \dots \in [U] $, window length $n$, approximate ratio $\epsilon \in (0, 1)$ and error rate $\delta \in (0,1)$, parameter $p$.

 \textbf{Output:}At each time $i$, output a
 $(1+\varepsilon)$-approximation to $F_p$ for window $[i - n, i]$.

 \textbf{Initialization:} $\beta \gets \left\lceil \log \frac{100 \max(p, 1)}{\varepsilon^{\max(p, 1)}} \right\rceil, \eta \gets \frac{\epsilon}{2 ^ {25}\cdot \log \frac 1{\epsilon}}, \gamma_j \gets 2^{3-j}$ for all $j \in [\beta]$, $\delta_{dif} = \frac{\delta}{2 ^ {25}\cdot \log \frac 1{\epsilon}}$

 \begin{enumerate}
     \item Let $\epsilon_0 = \frac 12$, prepare the matrix $\Pi_{top}$ for $(\frac{\epsilon_0 ^ p}{64}, \frac{\epsilon_0 ^ p}{128}, \delta_{SE}=\frac{\delta}{\log O(\log m)})$ \est to maintain a constant factor partition on the top level.

     \item Prepare the matrix $\pi$ for a \ind with dimension $O(\epsilon ^ 2 \log \frac 1{\delta})$.

     \item  To maintain the subroutine level, for each $j\in [\beta]$, prepare $C\cdot 2 ^ j\log n$ random matrices $\Pi_i \in d$ for a large constant $C$, each responsible for an $(\frac 18, \cdot \frac{2 ^ {-1/p} \epsilon}{2 ^ {32}}, \delta_{sub} = \frac{\delta}{O(C\epsilon ^ {-p} \log n)})$ \est. Store them in a queue $q_j$.

     \item For each level $j \in [\beta]$, generate the matrix $\Pi_{dif}$ for the $(\gamma_j, \eta, \delta_{dif})$ -\diff, with dimension $d = O(\frac{\gamma_j ^ {2/p}}{\eta ^ 2} (\log \frac 1{\eta} + \log \frac 1{\delta_{dif}}))$.
 \end{enumerate}

\noindent\makebox[\linewidth]{\rule{\paperwidth-10cm}{0.4pt}}

\begin{algorithmic}[1]
\For{each update $a_i$}
    \State \textbf{NewSW}($i$) \Comment{Create new subroutines for each update}
    \State \textbf{MergeSW}($i$) \Comment{Removes extraneous subroutines}
    \State Output \textbf{StitchSW}
\EndFor
\end{algorithmic}
\end{algorithm}

\begin{algorithm}[H]
\begin{algorithmic}[1]
\Function{NewSW}{$i$}
\State Let $s$ be the number of copies in the histogram, $t_j$ be the starting time of the $j$-th copy.
\State $t_{s+1} \gets i$.
\State Start a new top level \est using matrix $\Pi_{top}$.
\For{$j \in [\beta]$} \Comment{Start instances of each granularity for each copy }
    \State $t_{s + 1, j, 0} \gets i$
    \State Start a new instance of ($\gamma_j, \eta, \delta_{dif}$) - \dif $\textbf{SDiffEst}(s + 1, j, 0)$ 
    \State Start a new instance of \est $f_{s + 1, j, 0}$, using the matrix from top$(q_j)$, pop($q_j$).
    \State Use procedure \ref{rounding procedure} to round the \dif ending at time $i$. 
\EndFor
\EndFunction
\end{algorithmic}
\end{algorithm}

\begin{algorithm}[H]
\begin{algorithmic}[1]
\Function{MergeSW}{$t$}

\State Use Algorithm \ref{alg:reg1} with $\epsilon = \frac 12$ to maintain a constant factor partition at the top level, and an \ind along with each timestamps. The only difference is when deleting sketch or adding new sketch, we use Procedure \ref{rounding procedure} to calculate the result of rounding.
\For{$i \in [s], j \in [\beta]$} \Comment{Difference estimator maintenance}
    \State Let $r$ be the number of time steps $t_{i, j, *}$
    \For{$k \in [r-1]$} \Comment{Merges two algorithms with "small" contributions}    
        \If{$\hat f_{i,j,k-1} ^ {(t_{i, j, k-1}, t_{i, j, k+1})} \le 2^{(-j-14)/p} \cdot \hat f_{i,j,k-1} ^ {(t_{i, j, k - 1},t)}$}
            \If{$t_{i, j, k} \notin \{t_{i, j-1, *}\}$}
                \State Delete timestamp $t_{i, j, k}$ and the related \est. 
                \State Push back the matrix $\Pi$ to the queue $q_j$.
                \State Merge (add) the sketches for $\operatorname{SDiffEsT}(i, j, k -1)$  and $\operatorname{SDiffEsT}(i, j, k)$. 
                \State Use Procedure \ref{rounding procedure} to merge the rounding result.

            \EndIf
        \EndIf
    \EndFor
\EndFor
\EndFunction
\end{algorithmic}
\end{algorithm}

\subsection{Rounding Techniques}
In this subsection, we present several space-saving techniques that allow us to maintain correctness guarantees while significantly reducing the storage requirements of the sketching data structures. First, we describe how rounding can be applied to the \est sketches $f_{i,j,k}$ to reduce the number of bits required per entry. Next, we address the challenge of bounding the number of timestamps in the absence of a global union bound, introducing a rotating random set technique that probabilistically limits the number of active sketches. Finally, we apply a fine-grained rounding scheme to compress both the \ind and \dif sketches, ensuring that their space usage remains within our overall complexity bounds.

\paragraph{Storing the Sketch via \est $f_{i, j, k}$}

We now describe how to apply rounding to the estimates computed by \est to achieve space efficiency.

Throughout the algorithm, the sketch $f_{i,j,k}$ is used to estimate the $\ell_p$ norm over the interval from timestamp $t_{i,j,k}$ to either another timestamp $t'$ or the current time. While we maintain exact estimates for intervals ending at the current time, we observe that for estimates ending at previous timestamps $t'$, we can round the values to save space. Specifically, we round $f_{i,j,k}^{(t_{i,j,k}, t')}$ to the nearest power of $(1 + \frac{1}{16})$. The following claim shows that this rounding preserves the \est property up to a small additive loss:

\begin{clm} Let $u_1, \dots, u_m$ be a stream and $[l, r]$ an arbitrary window. Suppose $f$ satisfies the $(\epsilon_1, \epsilon_2)$ \est property on $[l, r]$. Define $\hat{f}^{(a,b)} = \text{Round}(f^{(a,b)})$ to the nearest power of $1 + \frac{1}{16}$. Then $\hat{f}$ satisfies the $(\epsilon_1 + \frac{1}{8}, (1+\frac{1}{16}) \epsilon_2)$ \est property on $[l, r]$. \end{clm}

Thus, in our analysis, since we use $(\frac{1}{8}, \frac{2^{-1/p} \epsilon}{2^{32}}, \delta_{\text{sub}})$-\est, we may safely treat the rounded sketches as $(\frac{1}{4}, \frac{2^{-1/p} \epsilon}{2^{30}}, \delta_{\text{sub}})$-\est without loss of generality.

However, directly storing $O(s^2)$ rounded values (where $s = O(\epsilon^{-p} \log n)$ is the number of timestamps) would still be space-inefficient. To address this, we introduce the following compression technique:

\begin{tech}\label{rounding tech} For each timestamp $t_{i,j,k}$, we maintain an array of length $O(\log n)$. The $c$-th entry stores the rank of the earliest timestamp $t'$ such that $f_{i,j,k}^{(t_{i,j,k}, t')} \ge (1 + \frac{1}{16})^c$. \end{tech}

A potential issue with this scheme is that the value $f_{i,j,k}^{(t_{i,j,k}, t_u)}$ might be "overestimated" due to a later timestamp $t_v$ with $f_{i,j,k}^{(t_{i,j,k}, t_v)} > f_{i,j,k}^{(t_{i,j,k}, t_u)}$. However, this does not violate the \est property, since $f_{i, j, k} ^ {(t_{i, j, k}, t_u)} < f_{i, j, k} ^ {(t_{i, j, k}, t_v)} < (1+\epsilon_1) \ell_p ^ {(t_{i,j,k}, t_v)} + \epsilon_2 \ell_p ^ {(l,r)} < (1+\epsilon_1) \ell_p ^ {(t_{i,j,k}, t_u)} + \epsilon_2 \ell_p ^ {(l,r)}$. 

Therefore, this approximation remains valid, and the total space required to store all the \est sketches is reduced to $O(s \log n \log s)$ bits.
\paragraph{Challenges in Bounding the Number of Timestamps}

The analysis in \cite{woodruff2022tight} relies on applying a union bound over all sketches and all times, which introduces an additional $\log n$ factor in the space complexity. In contrast, we leverage the \est to eliminate the need for such a union bound. However, this also introduces new challenges.

Consider the current window $W = [r - n + 1, r]$, and let $l$ be the earliest index such that $F_p^{(l, r)} \le 2 F_p^W$. Conditioned on the \est satisfying its guarantee over $[l, r]$, we can show that the number of timestamps at the first level is bounded and provides reliable guidance for constructing \dif.

The complication arises because this guarantee only holds with probability $1 - \delta$, and without a global union bound, we cannot ensure that the number of timestamps is always bounded across time. For instance, if the same matrix $\Pi$ is reused across all \est instances, then in the unlikely case that these matrices fail simultaneously, the number of timestamps could grow unbounded.

To mitigate this, we introduce the following technique:

\begin{tech}[Rotating Random Sets] For each level $j \in [\beta]$, we maintain a queue $q_j$ that holds $C \cdot 2^j \log n$ independent random matrices $\Pi$, for a sufficiently large constant $C$.

Each \est instance $f_{i,j,k}$ uses a distinct matrix $\Pi$ from the queue that has not been used before. When a sketch is deleted, its corresponding matrix is returned to $q_j$, making it available for reuse. \end{tech}

\begin{lem}\label{rotating} Fix any time $t$. With probability at least $1 - \frac{1}{\text{poly}(n)}$, the queue $q_j$ is non-empty for all levels $j \in [\beta]$. This ensures that there are at most $O(2^j \log n)$ active timestamps at level $j$.  \end{lem}

Hence, by a union bound over all levels and time, with high probability, the total number of timestamps across all levels remains bounded by $O(\epsilon^{-p} \log n)$ throughout the entire stream.

\paragraph{Rounding the \ind and \dif} Our algorithm maintains only $O(\log n)$ timestamps at the top level. However, each of these timestamps stores an instance of \ind sketch, which uses $O(\epsilon^{-2} \log n)$ bits, and the total space for the \dif sketches at each level is also $\tilde{O}(\epsilon^{-2} \log n)$. To reduce the overall space to $\tilde{O}(\epsilon^{-p} \log^2 n+ \epsilon ^ {-2}\log n)$, we apply the rounding technique from \cite{rounding} to compress the sketch values.


\begin{lem} (Correctness of rounding \ind)\label{rounding ind}
    
    Let $\epsilon', \delta' \in (0, 1)$, and $[l, r]$ be a window, $\Pi \in \mathbb{R} ^ {d\times U}$ be the sketch matrix, $\gamma = (\frac{\epsilon'\delta'}{dm \log (Um)}) ^ C$, $\mathcal{P}(\Pi x ^ {(l, r)})$ be the sketch vector after rounding procedure \ref{rounding procedure} with parameter $\gamma$. 
    
    Then $\med(|\mathcal{P}(\Pi x ^ {(l, r)})_1|, \dots, |\mathcal{P}(\Pi x ^ {(l, r)})_d|) = \med(|(\Pi x ^ {(l, r)})_1|, \dots, |(\Pi x ^ {(l, r)})_d|) \pm \epsilon'\|x ^ {(l, r)}\|_p$ with probability $\ge 1- \delta'$.
\end{lem}

\begin{proof}
    The proof simply follows Lemma \ref{rounding lemma}.
\end{proof}

Setting the $\epsilon', \delta'$ in Lemma \ref{rounding ind} to be $\frac 14 \epsilon, \frac 14 \delta$, we can store the rounded vector in space $O(d \log (\frac 1\gamma \log n)) = O(d (\log \frac{1}{\epsilon \delta} + \log\log n))$. 

Next, we prove the correctness of rounding the \diff, for $p = 2$. The argument is the same as Lemma \ref{rounding ind}. 

\begin{lem}(Correctness of rounding the \dif) \label{rounding dif}
 Let $\epsilon_0, \delta_0 \in (0, 1), q = 3$, $u, v$ be 2 frequent vector, $A \in \mathbb{R} ^ {d\times U}$ be the sketch matrix, $\gamma = (\frac{\epsilon_0  \delta_0}{dm \log (Um)}) ^ C$, for $p \in (0, 2)$, recall the way \dif estimate the difference between $\|u+v\|_p ^ p$ and $\|v\|_p ^ p$ is by calculate the mean of $\prod_{j=  q(i-1) + 1} ^ {qi} |(Au+Av)_j| ^ {p/q} - \prod_{j=  q(i-1) + 1} ^ {qi} |(Av)_j| ^ {p/q}$.  Then the mean of $\prod_{j=  q(i-1) + 1} ^ {qi} |(\mathcal{P}(Au+Av))_j| ^ {p/q} - \prod_{j=  q(i-1) + 1} ^ {qi} |(\mathcal{P}(Av))_j| ^ {p/q}$ only deviate $O((\epsilon_0 \frac{d ^ 2}{\delta_0 ^ 2}) ^ {q/p})\cdot F(v)$ from the original estimation with probability $\ge 1 - 2\delta_0$. 
    
\end{lem}

Thus, setting $\epsilon_0 = O(\frac{\epsilon \delta ^ 2}{d ^ 2})$ and $\delta_0 = O(\delta)$, we can store the each \dif in space $O(d (\log \frac {1}{\epsilon\delta} + \log\log n))$. 
\subsection{Nearly Optimal $F_p$ Estimation}

Now we introduce our main result for $F_p$ estimation.

\begin{theorem}\label{main est}
    Let $p\in [1, 2]$, for any fixed window $W$, with probability $\ge \frac 23$, Algorithm \ref{alg2} will give an $\epsilon$ approximation of $F_p ^ {W}$, and  use $\tilde{O}(\epsilon ^ {-p}\log ^ 2n + \epsilon ^ {-2} \log n)$ bits of space, more specifically, $O((\epsilon ^ {-p} \log ^ 2n + \epsilon ^ {-2}\log n \log ^ 4 \frac 1{\epsilon}) (\log \frac{1}{\epsilon} + (\log \log n) ^ 2))$.
\end{theorem}

Let the query window be $W=[r-n,r]$, and let $l$ be the smallest index that $F_p ^ {(l,r)} \le 2 F_p ^ W$, and $W'=[l,r]$. First, we prove that our top level \est gives a constant approximation. The proof can be easily seen from Lemma \ref{Eps Approx}. 

\begin{lem}\label{A.1}
 (Constant factor partitions in top level). 

With probability $\ge 1 - \delta$, $F_p ^ W \leq F_p ^ {\left(t_1, r\right)} \leq 2 F_p ^ W$ and $\frac{1}{2} F_p ^ W \leq F_p ^ {\left(t_{2}, r\right)} \leq F_p ^ W$.

\end{lem}

\begin{proof}

As we use a $(\frac {\epsilon ^ p}{64}, \frac {\epsilon ^ p}{64}, \delta_{SE})$ \est for $\epsilon = \frac 12$. By Lemma $\ref{Eps Approx}$, $\ell_p ^ {(t_2, r)} \ge \frac 34 \ell_p ^ {(t_1,r)}$. So $F_p ^ {(t_1, r)} \le (\frac 43) ^ p F_p ^ {(t_2, r)} \le 2 F_p ^ {W}$ and $F_p ^ {(t_2, r)} \ge (\frac 34) ^ p F_p ^ {(t_1, r)} \ge \frac 12 F_p ^ W$.
\end{proof}

Let $\mathcal{E}$ be the event that the top level \est gives a constant factor partition, and for the subroutine levels, all the $O(\epsilon ^ {-p} \log n)$ instances of \est satisfy the $(\frac 12, \frac {2 ^ {-1/p}\cdot \epsilon}{2 ^ {30}})$ \est property on window $W'$. In the subroutine level, we use $(\frac 12, \frac {2 ^ {1/p} \cdot \epsilon}{2 ^ {30}}, \delta_{Sub}< \frac {\delta}{O(\epsilon ^ p \log n)})$ \estt, so by a union bound, all the $O(\epsilon ^ p \log n)$ instances of \est satisfy the \est property in window $W'$ with probability $\ge 1 - \delta$. 

So by adjusting the constant for  $\delta$, with probability $\ge 1 - \delta$, $\mathcal{E}$ will happen.

We now show an upper bound on the moment of each substream whose contribution is estimated by the difference estimator, i.e., the difference estimators are well-defined.

\begin{lem}\label{A.2}(Upper bounds on splitting times)
    Let $p \in[1,2]$. Conditioned on $\mathcal{E}$, we have that for each $j \in[\beta]$ and all $k$, either $t_{1, j, k+1}=t_{1, j, k}+1$ or $F_p ^ {\left(t_{1, j, k}, t_{1, j, k+1}\right)} \leq 2^{-j-7} \cdot F_p ^ W$.
\end{lem}


     
     
     




\begin{lem}\label{A.3} 
     (Accuracy of difference estimators) (Lemma A.3 \cite{woodruff2022tight}) Let $p \in[1,2]$. Conditioned on the event in Lemma \ref{A.2} holding, we have that for each $j \in[\beta]$ and all $k$, 
     
     $\operatorname{SDiffEsT}\left(t_{i, j, k-1}, t_{i, j, k}, t, \gamma_j, \eta, \delta\right)$ gives an additive $\eta \cdot F_p\left(t_{i, j, k}, t\right)$ approximation to $F_p ^ {\left(t_{i, j, k-1}, t\right)}-F_p ^ {\left(t_{i, j, k}, t\right)}$ with probability at least $1-\delta$.
\end{lem}


\begin{lem}\label{A.4} (Geometric upper bounds on splitting times) (Lemma A.4 \cite{woodruff2022tight}) Let $p \in[1,2]$. Conditioned on the event in Lemma \ref{A.2} holding, we have that for each $j \in[\beta]$ and each $k$ that either $t_{i, j, k+1}=t_{i, j, k}+1$ or $F_p ^ {\left(t_{i, j, k}, m\right)}-F_p ^ {\left(t_{i, j, k+1}, m\right)} \leq 2^{-j / p-2} p \cdot F_p ^ {\left(t_{i, j, k+1}, m\right)}$.
\end{lem}



\begin{lem}\label{A.5}
    (Geometric Lower bounds on splitting times). Let $p \in[1,2]$, conditioned on $\mathcal{E}$, we have that for each $j \in[\beta]$ and each $k$ that
$$
F_p ^{\left(t_{i, j, k-1}, t_{i, j, k+1}\right)}>2^{-j-21} \cdot F_p ^ W .
$$
\end{lem}

    


    

\begin{algorithm}
\begin{algorithmic}[1]

\Function{StitchSW}{$t$}
\State $c_0 \gets t_1$
\State $X \gets g ^ {(t_1, t)}$ \Comment{By \ind, after rounding}
\For{$j \in [\beta]$} \Comment{Stitch sketches}
    \State Let $a$ be the smallest index such that $t_{1, j, a} \geq c_{j-1}$
    \State Let $b$ be the largest index such that $t_{1, j, b} \leq t - n + 1$
    \State $c_j \gets t_{1, j, b}$
    \State $Y_j \gets \sum_{k=a}^{b-1} \operatorname{SDiffEst}(1, j, k)$
\EndFor
\State \Return $Z := X - \sum_{j=1}^{\beta} Y_j$
\EndFunction

\end{algorithmic}
\end{algorithm}

\begin{lem}\label{A.6}(Number of level $j$ difference estimators) (Lemma A.6 \cite{woodruff2022tight})

Let $p \in[1,2]$, conditioned on the Geometric Lower bounds on splitting times to hold (Lemma \ref {A.5}) , we have that for each $j \in[\beta]$, that $k \leq 2^{j+23}$ for any $t_{1, j, k}$.

\end{lem}

    


We next bound the number of level $j$ difference estimators that can occur from the end of the previous level $j-1$ difference estimator to the time when the sliding window begins. We say a difference estimator $\mathcal{C}_{1, j, k}$ is active if $k \in\left[a_j, b_j\right]$ for the indices $a_j$ and $b_j$ defined in $\textbf{StitchSW}$. The active difference estimators in each level will be the algorithms whose output are subtracted from the initial rough estimate to form the final estimate of $F_p ^ W$.

\begin{lem}\label{A.7}(Number of active level $j$ difference estimators) (Lemma A.7 \cite{woodruff2022tight}) Conditioned on the Upper bounds and Geometric lower bounds on splitting times to hold (Lemma \ref{A.2} and \ref{A.5}) , for $p \in[1,2]$ and each $j \in[\beta]$, let $a_j$ be the smallest index such that $t_{1, j, a_j} \geq c_{j-1}$ and let $b_j$ be the largest index such that $t_{1, j, b_j} \leq r - n$. Then conditioned on $\mathcal{E}$, we have $b_j-a_j \leq 2 ^ {24}$.

\end{lem}

\begin{lem}\label{A.8}(Correctness of sliding window algorithm) For $p \in[1,2]$, Algorithm \ref{alg2} outputs a $(1\pm \varepsilon)$-approximation to $F_p ^ W$ with probability $\ge 1 - \delta$.
\end{lem}

Now we will prove the space bound.





\begin{lem}\label{final space}
    Let $\delta$ be a constant $\frac 13$, Algorithm \ref{alg2} uses space  $O((\epsilon ^ {-p} \log ^ 2n + \epsilon ^ {-2}\log n \log ^ 4 \frac 1{\epsilon}) (\log \frac{1}{\epsilon} + (\log \log n) ^ 2))$ bits. 
\end{lem}

Combining Lemma \ref{A.8}, which establishes the approximation guarantee, with the space bound given in Lemma \ref{final space}, we conclude the proof of Theorem \ref{main est}.

    



\section{Lower bounds}
\label{sec:lower_bounds}

We consider the following communication problem, which is a generalization of the classic GreaterThan communication problem.

\begin{problem}
\label{prob:multi_greater_than}
In the $\text{IndexGreater}(N,k)$ problem, Alice receives a sequence $(x_1, \ldots, x_k) \in [2^N]^k$. Bob receives an an index $j$ and a number $y \in [2^N].$   Alice sends a one-way message to Bob, and Bob must decide if $y > x_j$ or if $y \leq x_j.$
\end{problem}

\begin{lem}
\label{lem:index_greater}
Problem~\ref{prob:multi_greater_than} requires $\Omega(Nk)$ communication to solve with $2/3$ probability.
\end{lem}

This lower bound is known. For instance~\cite{braverman2018nearly} applies it in their lower bounds as well.  We supply a short proof in the appendix from the somewhat more standard Augmented Index problem.

\begin{lem}
\label{lem:main_lower_bound}
Let $p\in (1,2],$ and let $W, N, k, r$ be parameters satisfying the following:
\begin{itemize}
\item $\window \leq \unisize$
\item $2^r N k^{1 - 1/p} \leq \window^{1-1/p}$
\item $2^r Nk \leq \window$
\end{itemize}
Suppose that $\mathcal{A}$ is a streaming algorithm that operates on sliding windows of size $\window$, and gives a $1 \pm \frac{1}{12 k^{1/p}}$ multiplicative approximation to $F_p$ on any fixed window with $9/10$ probability, over the universe $[\unisize]$.

Then $\mathcal{A}$ uses at least $\Omega(kr\log N)$ space.
\end{lem}

\begin{proof}
See the appendix.
\end{proof}
\begin{theorem}
Fix $p \in (1,2]$. Suppose that $\mathcal{A}$ is a streaming algorithm that operates on a window of size $\window$, and outputs a $(1\pm \eps)$ approximation of $F_p$ on any fixed window with $9/10$ probability. Suppose that $\eps \geq n^{-1/p}.$ Then $\mathcal{A}$ must use at least \[
\Omega\left(\frac{(p-1)^2}{\eps^p}\log^2(\eps \min(\window,\unisize)^{1/p})
+ \frac{1}{\eps^2}\log(\eps^{1/p} \min(\unisize, \window))
\right)
\]space.  In particular for constant $p$ and $\window \geq \unisize,$ $\mathcal{A}$ must use at least \[\Omega\left(\frac{1}{\eps^p}\log^2(\eps \unisize) + \frac{1}{\eps^2}\log(\eps^{1/p} \unisize)\right)\] space.
\end{theorem}
\begin{proof}

The second term in the lower bound is shown by \cite{braverman2024optimality} in their Theorem 11 for general insertion-only streams.

For the first term, first suppose that $\window \leq \unisize.$ In Lemma~\ref{lem:main_lower_bound}, set $k = \frac{1}{\eps^p}$.  Then set $2^r$ and $N$ to both be $\Theta((\eps^{p-1}n^{1 - 1/p})^{1/2})$ so that the conditions of Lemma~\ref{lem:main_lower_bound} hold.  This gives the stated bound for $\window \leq \unisize.$ To handle the case where $\window \geq \unisize$, note that we can reduce from the $\window \leq \unisize$ setting; simply duplicate each item in the stream $\frac{\window}{\unisize}$ times.
\end{proof}


\bibliography{ref}
\newpage
\section{Appendix}
\subsection{Upper Bound Lemmas}
\paragraph{Proof of Lemma \ref{est}.}
\begin{proof}

Let $[l, r]$ be a window.

    Consider constructing a sequence of indices $l=t_0<t_1<\cdots<t_{q+1}=r$ inductively as follows: For each $1 \leq k \leq q+1$, define $t_k$ as the smallest index $\leq r$ such that $\ell_p^{\left(t_{k-1}, t_k\right)} \geq$ $\epsilon_2^{4 / p} \ell_p^{\left(l, r\right)}$. If no such index exists, set $t_{q+1}=r$. This ensures $q \leq \epsilon_2^{-4}$, because $\sum_{i =0} ^ q F_p ^ {(t_{i}, t_{i+1})} \ge q\epsilon_2 ^ 4 \cdot F_p ^ {(t_0, t_{q+1})} \ge F_p ^ {(t_0, t_{q+1})}$. 

Let $\Pi_i$ be the $i$-th row of $\Pi$. For all $1 \leq a<b \leq q+1$, define:

$$
l_{a, b}:=\#\left\{m: |\Pi_mx ^ {(t_a, t_b-1)}|<\left(1-\epsilon_1\right) \ell_p^{\left(t_a, t_b-1\right)}\right\}, \quad u_{a, b}:=\#\left\{m:|\Pi_mx ^ {(t_a, t_b-1)}|>\left(1+\epsilon_1\right) \ell_p^{\left(t_a, t_b-1\right)}\right\}
$$

By the property of distribution $\mathcal{D}_p$, we have $\mathbb{E}\left[l_{a, b}\right] \leq d\left(\frac{1}{2}-C \epsilon_1\right)$ and $\mathbb{E}\left[u_{a, b}\right] \leq d\left(\frac{1}{2}-C \epsilon_1\right)$ for a constant $C$. Let $S_{a, b}$ denote the event $\left\{l_{a, b} \leq \frac{1}{2}\left(d-C \epsilon_1\right)\right\} \cap\left\{u_{a, b} \leq \frac{1}{2}\left(d-C \epsilon_1\right)\right\}$.

By Lemma \ref{cher}, $\operatorname{Pr}\left[S_{a, b}\right] \geq 1-C^{\prime} \exp \left(-\Omega\left(d \epsilon_1^2\right)\right)-\exp (-\Omega(r))$. Choosing:

$$
d=\Theta\left(\epsilon_1^{-2}\left(\log \frac{1}{\epsilon_2}+\log \frac{1}{\delta}\right)\right), \quad r=\Theta\left(\log \left(\epsilon_2 \delta\right)^{-1}\right)
$$

we ensure $\operatorname{Pr}\left[S_{a, b}\right] \geq 1-\frac{\delta \epsilon_2^{-8}}{2}$. A union bound over all $\binom{q+1}{2}$ pairs $(a, b)$ guarantees that all $S_{a, b}$ hold simultaneously with probability $\geq 1-\frac{\delta}{2}$.

Next, define $E_{k, m}$ as the event that $\exists t \in\left[t_m, t_{m+1}-1\right]$ such that $\left|\left\langle\Pi_k, x^{\left(t_m, t\right)}\right\rangle\right|>\epsilon_2 \ell_p^{\left(l, r\right)}$. By Lemma \ref{prefix bound},

$$
\operatorname{Pr}\left[E_{k, m}\right] \leq C_p\left(\left(\frac{\epsilon_2^{4 / p}}{\epsilon_2}\right)^{\frac{2 p}{2+p}}+s^{-1 / p}\right)
$$

As $\epsilon_2 \leq \epsilon_1$ and setting $s=\Theta\left(\epsilon_1^{-p}\right)$, we have $\operatorname{Pr}\left[E_{k, m}\right] \leq \frac{C}{16} \epsilon_1$.
For each $m$, the probability that $\sum_{k=1}^d \mathbf{1}\left[E_{k, m}\right] \geq \frac{C}{8} d \epsilon_1$ is bounded by:

$$
\exp \left(-C^{\prime} d \epsilon_1^2\right)+\exp \left(-C^{\prime} r\right) \leq \delta \epsilon_2^{-8}
$$

For all intervals $[a, b]$, suppose $a \in\left[t_u, t_{u+1}\right)$ and $b \in\left[t_v, t_{v+1}\right)$ where $0 \leq u \leq v \leq q$. By triangle inequality we have:

$$
|\Pi_kx ^ {(a, b)}|=|\Pi_k x ^ {(t_u, t_v-1)}| \pm|\Pi_k x ^ {(t_u, a-1)}| \pm|\Pi_k x ^ {(t_v, b)}|
$$

We now count the number of $k$ such that
\[
|\Pi_k x ^ {(a, b)}| \leq (1 + \epsilon_1) \ell_p^{(a, b)} + \epsilon_2 \ell_p^{(l, r)}.
\]

For those $k$ such that $|\Pi_k x ^ {(t_u, t_v-1)}|\leq (1 + \epsilon_1) \ell_p^{(t_u, t_v-1)}$ and $|\Pi_k x ^ {(t_u, a-1)}|, |\Pi_k x ^ {(t_v, b)}| \leq \epsilon_2 \ell_p^{(l, r)}$, we have:
\[
|\Pi_k x ^ {(a, b)}| \leq (1 + \epsilon_1) \ell_p^{(t_u, t_v-1)} + 2 \epsilon_2 \ell_p^{(l, r)} \leq (1 + \epsilon_1) \ell_p^{(a, b)} + \left( (1 + \epsilon_1) \epsilon_2^{4/p} + 2 \epsilon_2 \right) \ell_p^{(l, r)}.
\]

By applying a union bound over all events and adjusting the constant factor of $\epsilon_2$, we can ensure:
\[
\# \left\{ k : |\Pi_k x ^ {(a, b)}| \leq (1 + \epsilon_1) \ell_p^{(a, b)} + \epsilon_2 \ell_p^{(l, r)} \right\} \geq \frac{d}{2}
\]

For the other side, we can use same method to bound 

\[
\# \left\{ k : |\Pi_k x ^ {(a, b)}| \ge (1 - \epsilon_1) \ell_p^{(a, b)} - \epsilon_2 \ell_p^{(l, r)} \right\} \geq \frac{d}{2}
\]

Thus, we conclude that, with probability $\ge 1 - \delta$, $\forall l \le a < b \le r$:
\[
\med_k\left( (\Pi x^{(a,b)})_k \right) \in (1 \pm \epsilon_1) \ell_p^{(a,b)} \pm \epsilon_2 \ell_p^{(l,r)}.
\]

\end{proof}

\paragraph{Proof of Lemma \ref{space}.}
\begin{proof}
    We refer to Lemma 11 in \cite{blasiok2017continuous}. Consider a sketch matrix $\Pi$ as in Construction \ref{construct est} - i.e. $\Pi \in \mathbb{R}^{d \times n}$ with random $\mathcal{D}_p$ entries, such that all rows are $r$-wise independent and all entries within a row are $s$-wise independent. Moreover let us pick some $\gamma=\Theta\left(\varepsilon m^{-1}\right)$ and consider discretization $\tilde{\Pi}$ of $\Pi$, namely each entry $\tilde{\Pi}_{i j}$ is equal to $\Pi_{i j}$ rounded to the nearest integer multiple of $\gamma$. The analysis identical to the one in \cite{kane2010exact} shows that this discretization have no significant effect on the accuracy of the algorithm, and moreover that one can sample from a nearby distribution using only $\tau=\mathcal{O}\left(\lg m \varepsilon^{-1}\right)$ uniformly random bits. Therefore we can store such a matrix succinctly using $\mathcal{O}(r s(\lg U+\tau)+r \lg d)$ bits of memory.

    Setting $d = O(\epsilon_1 ^ {-2}(\log \frac 1{\epsilon_2} + \log \frac 1{\delta}))$, $r = O(\log \frac 1{\epsilon_2} + \log \frac 1{\delta})$, and $s = O(\epsilon_1 ^ {-p})$ finishes the proof.
\end{proof}

\paragraph{Proof of Lemma \ref{add}.}
\begin{proof}
Construct another sequence of indices $1 = l_{s} < l_{s-1}  < \dots < l_0$, where $l_j$ is the lowest index with $\ell_p ^ {(l_j, r)} \le 2 ^ j$. Then with probability $\ge 1 - \delta \cdot \log r$,  the \est satisfies the \est property on window $[l_j, r]$ for all $j\in [s]$. So $\forall [a, b]$ that $1 \le a < l, l \le b \le r$, suppose $a \in [l_j, l_{j + 1}]$, by the \est property on window $[l_j, r]$, we have $f_p (x ^ {(a,b)}) = (1\pm \epsilon_1) \ell_p ^ {(a, b)} \pm \epsilon_2 \ell_p ^ {(l_j ^ {(i)},r)} = (1\pm \epsilon_1) \ell_p ^ {(a, b)} \pm \epsilon_2 (2\ell_p ^ {(a, b)} + 2\ell_p ^ {(b, r)})$.

\end{proof}

\paragraph{Proof of Lemma \ref{Eps Approx}.}
\begin{proof}

 Consider the case where \( t_1(i) + 1 < t_2(i) \). In this case, there must exist a time \( t \in (t_2(i), i] \) such that 
    \[
    f^{(t_2(i), t)} \geq \left( 1 - \frac{\epsilon^p}{32} \right) f^{(t_1(i), t)},
    \]
    
    using the fact that \( f \) is an $( \frac {\epsilon ^ p}{64} ,\frac {\epsilon ^ p}{128}, \delta')$ \estt. Therefore, by Lemma \ref{add}, with probability \( 1 - \delta' \cdot \log i \), for any interval \( [a, b] \) that $1 \le a < i - n, i - n\le b \le i$, we have 
    \[
    f^{(a, b)} = \left( 1 \pm \frac{\epsilon^p}{32} \right) \ell_p^{(a, b)} \pm \frac{\epsilon^p}{64} \ell_p^{W}
    \]

    And by the definition of \estt,  with probability $\ge 1 - \delta'$, for all interval $[a, b] \subseteq [i-n,i]$,

    \[
    f^{(a, b)} = \left( 1 \pm \frac{\epsilon^p}{64} \right) \ell_p^{(a, b)} \pm \frac{\epsilon^p}{64} \ell_p^{W}  = \ell_p^{(a, b)} \pm \frac{\epsilon^p}{32} \ell_p^{W}.
    \]
    
    With probability $\ge 1 - \delta$, both the above 2 cases will happen, this leads to the inequality:
    \[
    \left( 1 - \frac{\epsilon^p}{32} \right) \left( (1 - \frac{\epsilon ^ p}{64})\ell_p^{(t_1(i), t)} - \frac{\epsilon^p}{64} \ell_p^{W} \right) \leq \left( 1 - \frac{\epsilon^p}{32} \right) f^{(t_1(i), t)} \leq f^{(t_2(i), t)} \leq \ell_p^{(t_2(i), t)} + \frac{\epsilon^p}{32} \ell_p^{W}.
    \]

    So : 

    \[
(1 - \frac{\epsilon ^ p}{16})\ell_p ^ {(t_1(i), t)} - \frac{\epsilon ^ p}{16} \ell_p ^ W \le  \ell_p^{(t_2(i), t)}
    \]
    
By the super-additivity of \( F_p \) for \( p \geq 1 \), we obtain
\[
F_p^{(t_1(i), t_2(i))} \leq F_p^{(t_1(i), t)} - F_p^{(t_2(i), t)} \leq F_p^{(t_1(i), t)} - \left( \ell_p^{(t_1(i), t)} - \frac{\epsilon^p}{16} \ell_p^{W} - \frac{\epsilon^p}{16} \ell_p^{(t_1(i), t)} \right)^p.
\]
Simplifying further, we have
\[
F_p^{(t_1(i), t_2(i))} \leq F_p^{(t_1(i), t)} - \left( \ell_p^{(t_1(i), t)} - \frac{\epsilon^p}{8} \ell_p^{(t_1(i), i)} \right)^p \leq \frac{p \epsilon^p}{8} F_p^{(t_1(i), i)}.
\]
The last inequality follows from the fact that for \( a > b \), we have \( a^p - (a - b)^p \leq b \cdot p \cdot a^{p-1} \).

Now consider the case where \( a < b \), i.e., \( \ell_p^{(t_1(i), t)} \leq \frac{\epsilon^p}{8} \ell_p^{(t_1(i), i)} \). In this case, we can directly deduce that:
\[
\ell_p^{(t_1(i), t_2(i))} \leq \frac{\epsilon^p}{8} \ell_p^{W}.
\]

Thus, by the triangle inequality, we conclude that
\[
\ell_p^{(t_2(i), i)} \geq \ell_p^{(t_1(i), i)} - \ell_p^{(t_1(i), t_2(i))} \geq \left( 1 - \frac{\epsilon}{2} \right) \ell_p^{(t_1(i), i)}.
\]

\end{proof}

\paragraph{Proof of Lemma \ref{rounding lemma}}

\begin{proof}
If it is the first case, then we build a leaf node to represent the interval being rounded. If it is the second case, then we build a parent node, linked to the 2 children node which represent the 2 sub-intervals respectively. One can see that in the end that we will build a tree, with the node representing interval $[l, r]$ being the root. This is exactly a special case (where every operation is just round one sketch or merge 2 sketches) in Procedure \ref{distributed rounding procedure}, so the proof follows from Lemma \ref{rounding} and a union bound.
\end{proof}

\paragraph{Proof of Lemma \ref{rotating}.}
\begin{proof} 
We will prove, for each level $j$, the queue $q_j$ will never be empty (w.h.p), so there are at most $O(2 ^ j \log n)$ timestamps. Let the current time be $t$, and let $t_1 > t_2 > \dots > t_q$, where $t_i$ is the lowest index with $F_p ^ {(t_i, t)} \le 2 ^ i$. Then for each of the \est matrix $\Pi_i$, with probability $\ge 1 - \delta_{Sub} \cdot q \ge 1 - \delta$, it will satisfy \est property on all windows $[t_i, t]$. We refer to such a $\Pi_i$ as a "good sketch".

By a Chernoff Bound, the probability that there are $\ge \frac C2 \cdot 2 ^ j \log n$ good sketches in queue $q_j$ is $\ge 1 - \exp(-\Omega(C\cdot 2 ^ j \log n)) \ge 1 - \text{poly}(n)$ for a large constant $C$.

Suppose there are $s$ timestamps on the top level. From Lemma \ref {A.6}, there are at most $2 ^ {j+23}\cdot s = C'\cdot 2 ^ j \log n$ good sketches on level $j$. Choosing $C > 2C'$, we can infer that $q_j$ will not be empty.

\end{proof}

\paragraph{Proof of Lemma \ref{rounding dif}.}
\begin{proof}
    First we union bound on the events $\forall j, |(Au+Av)_j - \mathcal{P}(Au+Av))_j| \le \epsilon '\|u+v\|_p$ and $ |(Av)_j - \mathcal{P}(Av))_j| \le \epsilon'\|v\|_p$. We use the fact that for $Z \sim D_p$, we have $\Pr[|Z| > \lambda] \le O(\frac 1{\lambda ^ p})$. Picking $\lambda = O(\frac {2d}{\delta'})$, by a union bound, with probability $\ge 1 - \delta'$, $\forall j\in [d]$, $|(A(u+v))_j| \le \lambda \|u+v\|_p$ and $|(Av)_j| \le \lambda \|v\|_p$. 
    
    Conditioning on the 2 events both happening, we will have, $\forall i$,  
    
    $$\Big|\prod_{j=  q(i-1) + 1} ^ {qi} |(\mathcal{P}(Av))_j| ^ {p/q} - \prod_{j=  q(i-1) + 1} ^ {qi} |(Av)_j|^ {p/q} \Big|\le \prod_{j=  q(i-1) + 1} ^ {qi} (|(Av)_j|+\epsilon' \|v\|_p) ^ {p/q} - \prod_{j=  q(i-1) + 1} ^ {qi} |(Av)_j| ^ {p/q}$$

    Suppose $\epsilon'$ is sufficiently smaller than $\frac 1{\lambda ^ 2}$. Then, 
    
    $$\prod_{j=  q(i-1) + 1} ^ {qi} (|(Av)_j|+\epsilon' \|v\|_p) ^ {p/q} \le \prod_{j=  q(i-1) + 1} ^ {qi} |(Av)_j| ^ {p/q}+ (\epsilon' \|v\|_p) ^ {p/q} $$

    So 
    $$\prod_{j=  q(i-1) + 1} ^ {qi} (|(Av)_j|+\epsilon' \|v\|_p) ^ {p/q} -  \prod_{j=  q(i-1) + 1} ^ {qi} |(Av)_j| ^ {p/q} \le O((\epsilon' \lambda ^ 2) ^ {p/q} \|v\|_p)$$

    Then using the same argument for $u+v$, and using the fact that $\|u+v\|_p ^ p$ is bounded by a constant factor of $\|v\|_p ^ p$ when using \dif, we finish the proof.

\end{proof}

\paragraph{Proof of Lemma \ref{A.2}.}
\begin{proof}
     Suppose $t_{1, j, k+1} \neq t_{1, j, k}+1$. Then there is a time $t$, a timestamp that has been removed from the set tracked by the algorithm. At this time $t$, $f_{1,j,k} ^ {\left(t_{1, j, k}, t_{1, j, k+1}\right)} \leq 2^{(-j-14) /p}$. $f_{1, j, k} ^ {\left(t_1,t\right)}$.

     Conditioned on $\mathcal{E}$, $f_{1, j, k}$ satisfies the \est property within window $W'$. So all the sub-intervals $[a, b] \subseteq W'$ will satisfy 
     
     $$f_{1, j, k} ^ {(a,b)} \le \frac 32 \ell_p ^ {(a, b)} + \frac{\epsilon}{2 ^ {30}}\ell_p ^ {W} \le 4 \ell_p ^  W$$
     
     as $\ell_p ^ {(a, b)} \le \ell_p ^ {W'} \le 2 \ell_p^ W$, and 
     
     $$f_{1, j, k} ^ {(a, b)} \ge \frac 12 \ell_p ^ {(a, b)} - \frac{\epsilon}{2 ^ {30}}\ell_p ^ {W}$$

     Combining these 2 inequalities,
     \[
     \ell_p ^ {(t_{1, j, k}, t_{1, j, k + 1})} \le  2 f_{1, j, k} ^ {(t_{1,j,k}, t_{1, j, k + 1})} + \frac {\epsilon }{2 ^ {30}}\ell_p ^ W  \le 2 ^ {(-j-12) /p} f_{1,j,k} ^ {(t_{1, j, k}, t)} + \frac {\epsilon}{2 ^ {30}} \ell_p ^ W \le \Big(2 ^ {(-j - 8) / p} + \frac {\epsilon}{2 ^ {30}}\Big) \ell_p ^ W 
     \]

     Since $\frac {\epsilon ^ p}{2 ^ {30}} \le 2 ^ {-j-8}$, raising both sides to the power of $p$ yields $F_p ^ {(t_{1,j,k}, t_{1,j,k+1})} \le 2 ^ {-j-7} F_p ^ W$.

\end{proof}

\paragraph{Proof of Lemma \ref{A.5}.}
\begin{proof}
    Note that conditioned on $\mathcal{E}$, $f_{1, j, k-1}$ satisfy the \est property on window $W'$, then we have $\ell_p ^ {\left(t_{1, j, k-1}, t_{1, j, k+1}\right)} \geq \frac{1}{2} f_{1, j, k-1} ^ {(t_{1, j, k-1}, t_{1,j,k+1}) } - \frac {\epsilon}{2 ^ {30}} \ell_p ^ W$ . Since $\operatorname{MergeSW}$ did not merge the timestamp $t_{1, j, k}$ then it follows that 
    
    $$f_{1, j, k-1} ^ {\left(t_{1, j, k-1}, t_{1, j, k+1}\right)}>2^{(-j-14)/p} \cdot f_{1, j, k-1} ^ {\left(t_{1,j,k-1}, r\right)} \\ \ge 2 ^ {(-j-14)/p}\cdot \Big( \frac 12 \ell_p ^ {(t_{1,j,k-1},r) } - \frac{\epsilon}{2 ^ {30}}\ell_p ^ W\Big) \ge 2 ^ {(-j-18)/p} \cdot \ell_p ^ W$$

    So 

    $$\ell_p ^ {(t_{1, j, k-1}, t_{1, j, k+1})} \ge\Big( 2 ^ {(-j-20)/p} -  \frac {\epsilon}{2 ^ {30}}\Big) \cdot \ell_p ^ W$$
    
    Since $\frac {\epsilon ^ p}{ 2 ^ {30}} \le 2 ^ {-j-21}$, raising both sides to the power of $p$ finishes the proof.
\end{proof}

\paragraph{Proof of Lemma \ref{A.8}.}
\begin{proof}
Let $c_0=t_1$. Let $u$ be the frequency vector induced by the updates of the stream from time $t_1$ to $r$. For each $j \in[\beta]$, let $a_j$ be the smallest index such that $t_{1, j, a_j} \geq c_{i-1}$ and let $b_j$ be the largest index such that $t_{1, j, b_j} \leq r-n$. Since function \textbf{MergeSW} does not merge indices of $t_{i, j, k}$ that are indices of $t_{i, j-1, k^{\prime}}$, it follows that $t_{1, j, a_j}=c_{i-1}$. Let $u_j$ be the frequency vector induced by the updates of the stream from time $a_j$ to $b_j$. Let $v$ denote the frequency vector induced by the updates of the stream from $c_\beta$ to $r-n$. Thus, we have $u=\sum_{j=1}^\beta u_j+v+f$, so it remains to show that:
\begin{enumerate}
    \item $F_p(v+f)-F_p(f) \leq \frac{\varepsilon}{2} \cdot F_p(f)$.
    \item We have an additive $\frac{\varepsilon}{2} \cdot F_p(f)$ approximation to $F_p(v+f)$.
\end{enumerate}

By Lemma A.8 in \cite{woodruff2022tight}, the first statement is correct if we give a constant factor partition on the top level (Lemma \ref{A.1}) and the Geometric upper bound on splitting time property satisfies (Lemma \ref{A.4}). The second statement follows the property that the number of active level $j$ difference estimator is constant (Lemma \ref{A.7}). 

Thus, conditioned on $\mathcal{E}$ happening, the algorithm will give a $1\pm \epsilon$ approximation over $F_p ^ W$.  Recall $\mathcal{E}$ is the event that the top level \est gives a constant factor partition, and for the subroutine level, all the $O(\epsilon ^ {-p} \log n)$ \est satisfy the $(\frac 12, \frac {2 ^ {-1/p}\cdot \epsilon}{2 ^ {30}})$ \est property on window $W'$, and this happens with probability $\ge 1 - \delta$.

So, Algorithm \ref{alg2} outputs a $(1\pm \epsilon)$ approximation over $F_p ^ W$ with probability $\ge 1 - \delta$.

\end{proof}

\paragraph{Proof of Lemma \ref{final space}.}
\begin{proof}
First by Lemma \ref{rotating}, there will be at most $C\epsilon ^ {-p} \log n$ timestamps over all levels at all times, for a large constant $C$.    

The top level \est uses $\tilde{O}(\log ^ 2n)$ bits, and the \ind after rounding uses $O(\epsilon ^ {-2}\log n (\log \frac{1}{\epsilon} + \log \log n))$ bits.
    
    Let $s = O(\epsilon ^ {-p} \log n)$ be the total number of timestamps. Then storing \est using rounding technique \ref{rounding tech} uses $O(s\log n\log s) = O(\epsilon ^ {-p} \log ^ 2n (\log \frac 1{\epsilon} + \log \log n))$ bits of space. As the \est has dimension $d = O(\log \frac 1{\epsilon}+\log \log n)$, each timestamps has an accurate sketch, which costs $O(sd \log n) = O(\epsilon ^ {-p} \log ^ 2n (\log \frac 1{\epsilon} + \log \log n))$.     To store the $O(\epsilon ^ {-p} \log n)$ matrix of $(\frac 18, O(\epsilon), O(\frac 1{\epsilon ^ {-p}\log n}))$ \estt, by Lemma \ref{space}, it uses $O(\epsilon ^ {-p} \log ^ 2n (\log \frac 1{\epsilon} + \log \log n))$ bits of space. 

    For the \diff, for each level $j\in [\beta]$, there are $O(2 ^ j \log n)$ sketches, with $d = O(\frac{(2 ^ {-j}) ^ {2/p}}{(\epsilon / \log \frac 1{\epsilon}) ^ 2} \log \frac{1}{\epsilon})$. After rounding, each level uses $O(d\log n (\log \frac 1{\epsilon} + \log \log n))$ bits. Summing over all levels, the \dif uses $O(\epsilon^ {-2} \log n \log ^ 4\frac 1{\epsilon}(\log \frac 1{\epsilon} + \log \log n)  + d \log \frac 1{\epsilon}\log n)$ space. The last $O(d \log \frac 1{\epsilon}\log n)$ is because we store an accurate sketch at each level. To store the matrix of \diff, by Lemma \ref{dif lemma}, we use $O(d\log n(\log \log n) ^ 2)$ bits of space. 

    The space to store the \ind is no larger than \diff. 

\end{proof}

\subsection{Lower Bounds}

\paragraph{Proof of Lemma~\ref{lem:index_greater}}
\begin{proof}
We reduce from the Augmented Index problem on $Nk$ bits.  Suppose that in the AugmentedIndex problem Alice receives $Nk$ bits \[
a_{11}, \ldots, a_{1N}, a_{21}, \ldots, a_{2N}, \ldots, a_{k1}, \ldots, a_{kN}.
\]
and  Bob receives an index pair $(i,j)$ along with Alice's bits that occur prior to that index, and must output the corresponding bit.

Let $x_{\ell}$ be the number obtained by interpreting $a_{\ell 1} \ldots a_{\ell N}$ as a number written in binary.  Alice computes $x_1, \ldots, x_k$.

Bob forms the number $y$ whose binary expansion is $a_{i1} \ldots a_{i, (j-1)} 1 0 \ldots 0$ where there are $N - j$ zeros so that $y$ has $N$ binary digits.


Now Alice and Bob run the $\text{IndexGreater}(N,k)$ protocol, which allows Bob to decide with $2/3$ probability whether (i) $y > x_{i}$ or (ii) $y \leq x_{i}.$

Since $x_i$ and $y$ agree in the first $j-1$ digits, case (i) occurs precisely when $a_{ij} = 0$ and case (ii) occurs precisely when $a_{ij} = 1.$  Thus solving $\text{IndexGreater}(N,k)$ allows Bob to solve Augmented index on $Nk$ bits, which requires $\Omega(Nk)$ communication to solve with constant probability.

\end{proof}

\paragraph{Proof of Lemma~\ref{lem:main_lower_bound}.}
\begin{proof}
We reduce from Problem~\ref{prob:multi_greater_than} above.  Let 
\[
x^{(1)}_1, \ldots, x^{(1)}_k, x^{(2)}_1, \ldots, x^{(2)}_k, \ldots x^{(r)}_1, \ldots, x^{(r)}_k,
\]
with each $x^{(i)}_j \in [N]$
be the input to Alice in Problem~\ref{prob:multi_greater_than}.  Let Bob receive $y \in [N]$ along with an index pair $(\alpha, \beta).$ Bob must distinguish between $y > x^{(\alpha)}_{\beta}$ and $y \leq  x^{(\alpha)}_{\beta}.$

We will construct a stream by concatenating a series of blocks.  We first give the general structure of such a such a block.

For a sequence $x_1, \ldots, x_k \in [N]^k$, we define a $B$-block as follows.
The block is a sequence of length $B.$  Divide the block into $k$ consecutive sub-blocks of size $B/k$. Also divide each sub-block into $N$ consecutive mini-blocks each of size $\frac{B}{Nk}.$  In mini-block $x_i$ of sub-block $i$, we fill it with $\frac{B}{Nk}$ copies of a fixed universe item.  We populate the remaining unfilled entries of the block with distinct elements.  Note that the $F_p$ of a $B$-block is 
\begin{equation}
\label{eq:Fp_of_block}
B - \frac{B}{N} + \frac{B^p}{N^p k^{p-1}} \leq B + \frac{B^p}{N^p k^{p-1}}.
\end{equation}

Alice constructs a series of blocks $B^{(1)}, \ldots B^{(r)}$ where block $B^{(i)}$ is a $(\window/2^i)$-block for the sequence $x^{(i)}_1, \ldots x^{(i)}_k$ and concatenates them to form a stream $\mathcal{S}.$  For concreteness we  construct the stream such that the mini-block associated to $x^{(i)}_j$ contains repetitions of the universe item $u^{(i)}_j := ik + j.$ This is simply to ensure that the blocks are disjoint, and so that Bob knows the universe item associated with each index.

Alice runs the streaming algorithm on $\mathcal{S}$ and sends the resulting sketch to Bob.  Bob then computes $Q := \frac{\window}{2^i N k^{1 - 1/p}},$ and simulates adding $Q$ copies of $u^{(\alpha)_{\beta}}$ to the end of the stream, running the streaming algorithm to update the sketch.  Bob then appends new singletons to the stream so that the sliding window begins on mini-block $y$ of sub-block $\beta$ in block $\alpha.$

If $x^{(\alpha)}_{\beta} \geq y$, then by our construction, there are $\frac{\window}{2^i Nk}$ copies of $u^{(\alpha)_{\beta}}$ remaining in the stream.  Otherwise there are no copies of $u^{(\alpha)_{\beta}}$ remaining in the stream.  The $F_p$ of the resulting stream is at most $\window + \frac{\window^p}{2^{ip} N^p k^{p-1}}$ by summing Equation~\ref{eq:Fp_of_block} over the remaining blocks. 

In the first case, replacing $Q$ singletons with $Q$ copies of $u^{(\alpha)_{\beta}}$  increases the $F_p$ of the stream by 
\[
-Q + \left(Q + \frac{\window}{2^i Nk}\right)^p - \left(\frac{\window}{2^i Nk}\right)^p
\geq Q^p - Q + p\frac{Q^{p-1}\window}{2^i N k}.
\] In the second case, the $F_p$ of the stream increases by $Q^p - Q.$

Given an $(\eps, F_p)$-approximation approximation algorithm, we can run it twice to approximate the increase to within an additive error of 
\[
2\eps \left(\window + \frac{\window^p}{2^{ip} N^p k^{p-1}} + \left(Q + \frac{\window}{2^i N k}\right)^p \right)
\leq \frac{12\eps \window^p}{2^{ip} N^p k^{p-1}},
\]
where we used the assumption that $2^r N k^{1 - 1/p} \leq \window^{1 - 1/p}.$

In order to distinguish between the two cases, it suffices to have
\[
\frac{12\eps \window^p}{2^{ip} N^p k^{p-1}} \leq \frac{Q^{p-1} \window}{2^i N k},
\]
which after plugging in our choice of $Q$ is equivalent to
\[
\frac{12\eps \window^p}{2^{ip} N^p k^{p-1}} \leq \frac{\window^p}{2^{ip} N^p k^{1 + (p-1)^2/p}}.
\] 
This occurs for $\eps \leq \frac{1}{12 k^{1/p}}.$

\end{proof}

\end{document}